\documentclass[onecolumn]{elsart}

\usepackage{amsmath, amssymb, amsthm, dsfont, bm}
\usepackage{graphicx}
\usepackage{tabularx}
\usepackage{lineno}

\newcolumntype{C}[1]{>{\centering\arraybackslash}p{#1}}
\newtheorem{theorem}{Theorem}
\newtheorem{definition}{Definition}

\newtheorem{lemma}{Lemma}
\newtheorem{proposition}{Proposition}
\newtheorem{propositiona}{Proposition}

\newcommand {\R} {{\mathds R}} 
\newcommand {\N} {{\mathds N}} 
\newcommand {\C} {{\mathds C}}

\newcommand {\BE} {\begin{eqnarray}}
\newcommand {\EE} {\end{eqnarray}}
\newcommand {\BS} {\begin{eqnarray}}
\newcommand {\ES} {\end{eqnarray}}
\newcommand {\BEW} {\begin{eqnarray*}}
\newcommand {\EEW} {\end{eqnarray*}}

\newcommand {\x} {{\bm x}}

\renewcommand {\d} {{\rm d}}
\renewcommand {\phi} {{\varphi}}

\begin{document}

\begin{frontmatter}

\title{Tackling the Gross-Pitaevskii energy functional with the Sobolev gradient -- Analytical and numerical results}
  
\author[Paris]{Parimah Kazemi \corauthref{cor}},
\corauth[cor]{Corresponding author.}
\ead{kazemi@ann.jussieu.fr}
\author[Ulm]{Michael Eckart}

\address[Paris]{Laboratoire J.-L. Lions, Universit\'e Pierre et Marie Curie, F-75013 Paris, France}
\address[Ulm]{Institut f\"ur Quantenphysik, Universit\"at Ulm, D-89069 Ulm, Germany}

%\date{\today}

\begin{abstract}

In the first part of this contribution we prove the global existence and uniqueness of a trajectory that globally converges to the minimizer of the Gross-Pitaevskii energy functional for a large class of external potentials. Using the method of Sobolev gradients we can provide an explicit construction of this minimizing sequence. 

In the second part we numerically apply these results to a specific realization of the external potential and illustrate the main benefits of the method of Sobolev gradients, which are high numerical stability and rapid convergence towards the minimizer.

\end{abstract}

\end{frontmatter}

\pagebreak

%%%%%%%%%%%%%%%%%%%%%%%%%%%%%%%%%%%%%%%%%%%%%%%%%%%%%

\tableofcontents

\newpage

%%%%%%%%%%%%%%%%%%%%%%%%%%%%%%%%%%%%%%%%%%%%%%%%%%%%%
\section*{Introduction} 
%%%%%%%%%%%%%%%%%%%%%%%%%%%%%%%%%%%%%%%%%%%%%%%%%%%%%

Ever since Bose-Einstein condensation was realized in dilute bosonic gases in 1995 \cite{Anderson95,Davis95}, experimentalists as well as theoreticians have been interested in describing the experimental results with an accurate yet also efficient theoretical framework. The most common approach in this context is the Gross-Pitaevskii equation \cite{Gross61,Gross63,Pitaevskii61} which dates back to the 1960s and provides a meaningful description in the regime of weakly interacting bosons. Recently, it was also shown rigorously that the exact ground state of dilute, interacting bosons in an external trap can exactly be described by the minimizer of the Gross-Pitaevskii energy functional in the appropriate asymptotic limit \cite{Lieb00}. 

Even today it is of utmost importance to have an efficient way of finding the ground state of weakly interacting dilute bosons in an arbitrary external trap as this state is either of interest in itself or the starting point for a time evolution. There exists a large number of numerical approaches to find this ground state, among them the density-matrix renormalization group method \cite{whitedmrg}, the Multi-Configuration Time Dependent Hartree method \cite{Meyer90,Meyer00,zollner2006}, the Numerov method \cite{Gonzalez97}, the fast {Fourier} transform method \cite{Bulirsch} using steepest descent or imaginary time propagation, the finite element method \cite{Jung} or the method of Sobolev gradients \cite{Neuberger}. However they differ significantly in their numerical stability and efficiency and the extension to three-dimensional problems is often accompanied by technical difficulties or problems with the efficiency. In this contribution we present the method of Sobolev gradients in detail, since it provides a very powerful framework for analytical investigations and additionally offers the possibility of combining analytical as well as numerical results in a natural way. This method has before been applied, e.g. in the works of \cite{Garcia,Lookman,Garcia1}, yet our study differs from these previous works, since we obtain global convergence for the trajectory that leads to the minimizer. Using our finite differencing scheme, we are thus guaranteed convergence as the mesh goes to zero. Furthermore our results are applicable to a very general group of trapping potentials, both numerically and theoretically.

We first discuss the relevant properties of the Gross-Pitaevskii energy functional which allow us to prove the global convergence of the minimizer of this functional with the help of Sobolev gradients. Having shown this result, our numerical simulations are guaranteed to provide the correct ground state in the limit of an infinitely fine grid on which the Gross-Pitaevskii equation is solved. Additionally, we highlight the numerical advantages of the method of Sobolev gradients in the second part of this contribution, where we discuss one-, two- and three-dimensional simulations to obtain the ground state of the Gross-Pitaevskii energy functional in the presence of a generalized Mexican hat potential.

To set the stage we begin with the definition of the $d$-dimensional Gross-Pitaevskii energy functional $\tilde{E}(\psi)$ for $N$ interacting bosons of mass $m$ in an external trapping potential $\tilde{V}_{trap}(\tilde{\x})$
\BEW
\tilde{E}(\psi)&=&\int_{\R^d} \left(\frac{\hbar^2}{2m} | \nabla \psi(\tilde{\x})|^2 + \tilde{V}_{trap}(\tilde{\x})| \psi(\tilde{\x})|^2 + \frac{\tilde{g}}{2} | \psi(\tilde{\x})|^4\right) \d\tilde{\x} \, ,
\EEW
where $\tilde{g}$ denotes the coupling constant which for $d=3$ reads $\tilde{g}={4 \pi \hbar^2 a_s}/{m}$ and is determined by the s-wave scattering length $a_s$. The wave function $\psi(\tilde{\x})$ is complex valued on $\R^d$ and we are interested in minimizing the energy functional subject to the normalization condition 
\BEW
\int_{\R^d} | \psi(\tilde{\x})|^2 \d\tilde{\x} &=& N . 
\EEW
We restrict ourselves to external potentials that are measurable, locally {bound\-ed} and tend to infinity for $|\tilde{\x}| \to \infty$. Thus the potential is bounded from below and without loss of generality we consider the minimum to be zero. Neglecting the quartic self-interaction term in the energy functional, the minimizer is equivalent to the ground state of the standard Schr\"odinger Hamiltonian
\BEW
\hat{H} & = & -\frac{\hbar^2}{2m} \Delta + \tilde{V}_{trap}(\tilde{\x})
\EEW
which provides a natural energy unit of $\hbar \omega$. For the most commonly used example of an isotropic harmonic oscillator $\tilde{V}_{trap}(\tilde{\x})=m\omega^2\tilde{\x}^2/2$ this can be shown trivially and the corresponding ground state wave function is a Gaussian. From the energy unit we can derive an appropriate length unit which is given by $a_0=\sqrt{{\hbar}/{m \omega}}$. By rescaling the original energy functional in terms of these units, that means $\tilde{\x} = a_0 {\x}$ and $\tilde{E}(\psi) = \hbar \omega E(u)$, we obtain
\BE
 \label{gpnondim}
E(u)=  \int_{\R^d} \left(\frac{1}{2}| \nabla u(\x) |^2 + V_{trap}(\x) |u(\x)|^2 + \frac{g}{2} |u(\x)|^4\right) \d\x \,,
\EE
where $u(\x)= {{a_0}^{d/2}}\psi(\tilde{\x})$, $V_{trap}(\x) = \tilde{V}_{trap}(\tilde{\x})/{\hbar\omega}$ and $g = {a_0}^{-d}{\tilde{g}}$. The normalization condition remains unchanged and reads
\BE\label{normcond}
\int_{\R^d} | u(\x)|^2 \d\x = N \,. 
\EE

\pagebreak
For the following discussion of the minimization of this functional, we restrict ourselves to a bounded domain $D \subset \R^d$
%and treat $H^{1,2}(D,\C)$ as $H=H^{1,2}(D) \times H^{1,2}(D)$. For \mbox{$u \in H$} and $u=(r,s)$, we denote the first order partial derivatives of $r$ and $s$ by
%\BEW
%\nabla u=\left(\begin{array}{c}
%r_1,s_1\\r_2,s_2\\ \vdots \\r_d,s_d\end{array}\right) \, .
%\EEW
and minimize the energy functional \eqref{gpnondim} subject to the constraint \eqref{normcond}. This problem has already been investigated in many studies, both theoretical and numerical (\cite{Fetter1995,Edwards1996,Schneider1999,Chiofalo2000,Bao,Bao1,Minguzzi2004} and references therein). The existing result for convergence is the existence of a minimizing sequence that converges strongly in the $L^p$ norm to a member of the Sobolev space $H=H^{1,2}(D,\C)$.  In contrast, we will show via steepest descent with a Sobolev gradient the strong convergence in the $H$ norm. This is our main theoretical result and the convergence is obtained for a wide class of trapping potentials. Furthermore, we do not only prove the convergence in an abstract sense but actually derive the minimizing sequence $z(t)$ in a constructive way, which allows us to use this sequence for our numerical simulations.

The method of using Sobolev gradients to obtain stationary solutions of an energy functional provides various advantages compared to other methods which are currently being used. First, since the convergence is in the $H$ norm we not only get that $z(t)$ converges in the $L^2$ norm to a member of $H$, but we also get that the partial derivatives $\partial_i (z(t))$ converge in the $L^2$ norm, for $i=1,2,\ldots,d$.  Not only is this a stronger form of convergence than what has been proved so far, but it also demonstrates that the numerical simulations based on this method will be well behaved in two ways. Since the continuous version of steepest descent converges, it is clear that for an increasing number of discretization points in our numerical simulations, the solution of the discrete minimization problem converges to a member of $H$.  Additionally, one can expect that the convergence of the discrete problem will be smooth since the convergence of the continuous problem is in $H$.  We outline these benefits of the Sobolev gradient in more detail in the numerical part.  Eventually it is the combination of our theoretical results and their practical numerical implementation which constitutes the main new contribution of this paper to the problem of finding and analyzing critical points of the Gross-Pitaevskii energy functional.

%%%%%%%%%%%%%%%%%%%%%%%%%%%%%%%%%%%%%%%%%%%%%%%%%%%%%
\section{Analytical results} 
%%%%%%%%%%%%%%%%%%%%%%%%%%%%%%%%%%%%%%%%%%%%%%%%%%%%%

The goal of the analytical part of our contribution is the proof of the global existence and uniqueness of a trajectory that globally converges to the minimizer of the Gross-Pitaevskii energy functional. To arrive at this goal, section \ref{basic_properties} provides basic properties of this functional which are essential for our proof.
%of thuse the Sobolev embedding theorem to show that the energy functional is uniformly and strictly convex and that $u$ is a minimizer of $E$ if and only if $u$ is a minimizer of $E_{\epsilon}$ where $E_{\epsilon}(u) = E(u) + \epsilon \|u\|_{L^2}^2$.  
In section \ref{minimization_Sob} we define the gradient that we will use for the minimization and describe how we incorporate the normalization constraint of the desired ground state into this gradient using a particular class of projections.  We set up a trajectory using this gradient and consider the resulting evolution equation. 

Of the final two sections of the analytical part, section \ref{proj} deals with properties of the class of projections we obtained. We will need these properties in section \ref{steepest_descent} in order to obtain global existence and uniqueness as well as global convergence for the minimizing trajectory.

To avoid confusion, we define the notion of a Fr\'{e}chet derivative before we start with the analytical details, since it will commonly be used throughout this paper.  If $G$ is a function from one Hilbert space to another, we denote the first Fr\'{e}chet  derivative of $G$ at $u$ by $G'(u)$ and by $G''(u)$ we denote the second Fr\'{e}chet derivative of $G$ at $u$, provided these derivatives exist.  Furthermore, we remind the reader of the definition of Fr\'{e}chet differentiability.
\vspace{3mm}
\begin{definition}
Let $G:H_1 \rightarrow H_2$ be a map between the two Hilbert spaces $H_1$, $H_2$.  $G$ is Fr\'{e}chet differentiable at $u \in H_1$ if there is a continuous linear transformation $T_u: H_1 \rightarrow H_2$ so that for all $\epsilon > 0$ there exists $\delta > 0$ so that for all $h \in H_1$ with $\|h\|_{H_1} < \delta$ the following inequality holds
\BEW
\frac {|E(u+h) - E(u) - T_u h|}{\|h\|_{H_1}} < \epsilon \, .
\EEW
We write $T_u=G'(u)$ and $G$ is twice Fr\'{e}chet differentiable if the map $u \rightarrow E'(u)h$ from $H_1$ to $H_2$ is Fr\'{e}chet differentiable.
\end{definition}

In cases where $G$ is known to be Fr\'{e}chet differentiable, the Fr\'{e}chet derivative and the Gateaux derivative of $G$ coincide. The Fr\'{e}chet derivative can thus be computed in the following way
\BEW
G'(u)h = \lim_{t \rightarrow 0} \frac{G(u + th) - G(u)}{t}
\EEW
and is similar to a functional derivative which is commonly used in Quantum Field Theory.

\subsection{Basic properties of the energy functional}
\label{basic_properties}
%%%%%%%%%%%%%%%%%%%%%%%%%%%%%%%%%%%%%%%%%%%%%%%%%%%%%

In this section we discuss the basic properties of the functional \eqref{gpnondim} which are necessary to obtain our analytical results, in particular the size and shape of the domain $D$ of the energy functional as well as its convexity.  Furthermore, we state the assumptions we make on the trapping potential and show that they do not lead to a loss of generality. 

We note that the functional $E(u)$ is well defined for potentials that are non-negative and bounded on $D$, which follows from the Sobolev embedding theorem. Here we assume that $D$ is a bounded domain in $\R^d$ with a regular boundary, that is it satisfies the cone condition. In practice, $D$ is usually determined by the condition that the ground state effectively vanishes at the border of the domain. In the case of an isotropic harmonic trap as external potential, that is $V_{trap}(\x)={\x^2}/{2}$, we can identify the domain to be a ball of radius $R$. This radius is a multiple of the width of the Gaussian in the case of weakly interacting bosons, whereas it is a multiple of the Thomas-Fermi radius in the case of a strongly interacting BEC. For the following discussion let $L^2=L^2(D,\C)$ and $H=H^{1,2}(D,\C)$.

For three dimensions, the Sobolev embedding states that 
\vspace{3mm}
\begin{theorem}
$H^{1,2}(D,\R)$ is continuously embedded in $L^q(D,\R)$ if $D$ is open and satisfies the cone condition and $2 \leq q \leq 6$. 
\end{theorem}

Furthermore, one can obtain that $H^{1,2}(D,\C)$ is continuously embedded in $L^q(D,\C)$ if $D$ is open and satisfies the cone condition and $q$ satisfies the inequality for the above theorem.  For the one- and two-dimensional case this theorem also holds, but a much stronger result is available.   See \cite{adams} for details on Sobolev spaces and the Sobolev embedding theorem.  Next, we present our result for uniform convexity.
\vspace{3mm}
\begin{lemma}
\label{convexity}
If we assume that the trapping potential is bounded away from zero, $E(u)$ is uniformly and strictly convex.  In other words, for all $\epsilon > 0$ with $V_{trap}(\x) \geq \epsilon$, there exists $\delta > 0$ so that for all $u$ and $h \in H$
\BEW
E''(u)(h,h) &\geq& \delta \|h\|_H^2 \, .
\EEW
\end{lemma}

\begin{proof}
To see this we do a direct computation to obtain
\begin{eqnarray*}
E''(u)(h,h) & = & \int_D \left(|\nabla h(\x)|^2 + 2V_{trap}(\x) |h(\x)|^2 \right) \d\x\\
&&+\int_D 2g |h(\x)|^2|u(\x)|^2 + 4g \ (\Re \langle h(\x) ,u(\x) \rangle)^2 \d\x \\
& \geq & \int_D \left(| \nabla h(\x)|^2 + 2 V_{trap}(\x) |h(\x)|^2\right) \d\x \\
& \geq &  \int_D \left(\delta | \nabla h(\x)|^2 + \delta |h(\x)|^2\right) \d\x = \delta \|h\|_H^2
\end{eqnarray*}
where $\delta = \text{min}(2\epsilon,1)$.
\end{proof}

Furthermore, note that this assumption on the trapping potential does not lead to a loss of generality as given by the following lemma.

\pagebreak

\begin{lemma}
Let $E(u)$ be as in \eqref{gpnondim} and
\BEW
E_{\epsilon}(u) &=& E(u) + \epsilon \int_D |u(\x)|^2\d\x \,.
\EEW
Then for $\beta(u)= \int_D |u(\x)|^2\d\x$ and $h \in null(\beta'(u))$, $E'(u)h=0$ iff $E_{\epsilon}'(u)h=0$.
\end{lemma}

\begin{proof}
Note that 
\BEW
E_{\epsilon}'(u)h & = & E'(u) h + \epsilon \beta'(u)h.
\EEW
Thus if $h \in null(\beta'(u))$ we have that $\beta'(u)h=0$ and
\BEW
E_{\epsilon}'(u)h=E'(u)h + \epsilon \beta'(u)h = E'(u)h.
\EEW
Thus we have that $E_{\epsilon}'(u)h=0$ iff $E'(u)h=0$ for $h \in null(\beta'(u))$.  
\end{proof}

Thus if the minimum of the trapping potential is zero, we can add a multiple of $\int_D |u(\x)|^2\d\x$ to the energy to obtain $E_{\epsilon}(u)$ as given above.  Additionally, we see that if we have $u$ so that $E_{\epsilon}'(u)h=0$ for all $h \in null (\beta'(u))$, then $E'(u)h=0$ also for all $h \in null(\beta'(u))$ and the desired result is achieved.

%%%%%%%%%%%%%%%%%%%%%%%%%%%%%%%%%%%%%%%%%%%%%%%%%%%%%
\subsection{Minimization of the energy functional with the Sobolev gradient}
\label{minimization_Sob}
%%%%%%%%%%%%%%%%%%%%%%%%%%%%%%%%%%%%%%%%%%%%%%%%%%%%%

In this section we present the method of Sobolev gradients which can be applied to the problem of finding the minimizer of the Gross-Pitaevskii energy functional \eqref{gpnondim}. We define the Sobolev gradient and subsequently discuss how the normalization constraint \eqref{normcond} can be incorporated in the gradient. We pursue this approach since we propose a direct minimization for which a traditional Lagrange multiplier method is not suited. The motivation and background for this minimization closely follows \cite{jwn}.  

Since $E$ is a continuously twice differentiable function from $H$ to $\R \subset \C$, the Riesz representation theorem provides the result that for $u \in H$, there is a member of $H$ which we denote by $\nabla_H E(u)$ so that
\BE \label{sblvgrad}
E'(u)h= \langle h , \nabla_H E(u) \rangle_{H}
\EE
for all $h \in H$. We denote $\nabla_H E(u)$ to be the Sobolev gradient of $E$ at $u$.  

We wish to minimize $E$ using the method of steepest descent with this gradient. However, we first need to incorporate the normalization constraint into our formulation.  For $u \in H$, let 
\BE \label{beta}
\beta(u) = \int_D |u(\x)|^2\d\x \, ,
\EE
which means that we want to minimize the energy functional \eqref{gpnondim} subject to the constraint $\beta(u)=N$.  We propose a direct minimization method using steepest descent with a gradient obtained with respect to a Sobolev inner product.  In order to incorporate the normalization constraint into our formulation, we cannot use a traditional Lagrange multiplier method to add a multiple of the constraint to the energy as this would make the resulting functional unbounded from below.  

However, in drawing motivation from the method of Lagrange multipliers, we see that if such a minimum is achieved at $u$ then the gradient of the energy functional $E$ at $u$ and the gradient of $\beta$ at $u$ are parallel, and additionally $\beta(u)=N$.  Thus we seek $u \in H$ so that $\beta(u)=N$ and if $\langle h , \nabla_H \beta(u) \rangle_H = 0$ for $h \in H$, then $\langle h , \nabla_H E(u) \rangle_H=0$ also.  Note that this translates into finding $u$ so that $\beta(u)=N$ and $E'(u)h=0$ for all $h \in null (\beta'(u))$.  The rest of this section is devoted to a formulation of how such a $u$ can be found in the scope of steepest descent.  The general setup on how to include constraints in the gradient is explained in detail in chapter ``Boundary and Supplementary Conditions'' of \cite{jwn}, but the essential ingredients are reviewed in the next paragraph.

As previously stated, for each $u \in H$ one can find a member of $H$, denoted by $\nabla_H E(u)$, so that equation \eqref{sblvgrad} holds for all $h \in H$. Note that for each $u \in H$, the nullspace of $\beta'(u)$ is a closed linear subspace of $H$.  Thus for each $u \in H$, there is a projection from $H$ onto the nullspace of $\beta'(u)$.  We denote this projection by $P_u$ in this paper.  Let $u_0 \in H$ so that $\beta(u_0)=N$ and define $z: [0, \infty) \rightarrow H$ so that
\BE \label{sd}
z(0)=u_0 \text{ and } z'(t)=-P_{z(t)} \nabla_H E(z(t)).
\EE
Here $z'(t)$ denotes the Fr\'{e}chet derivative as defined above and we will show that $z$ is defined for all $t \geq 0$ in section \ref{steepest_descent}. Essentially, we project the Sobolev gradient of $E$ at $z(t)$ onto the nullspace of $\beta'(z(t))$ since this leads to a constant $\beta(z)$. This can be seen by considering
\BEW
(\beta (z))'(t) = \beta'(z(t)) z'(t) = - \beta'(z(t)) (P_{z(t)} \nabla_H E(z(t)))= 0 
\EEW
where we used that $P_{z(t)}$ is the projection of $H$ onto the nullspace of $\beta'(z(t))$. Hence, $\beta (z)$ is constant and if $u=\lim_{t \rightarrow \infty} z(t)$, then $\beta(u) = \beta(u_0)$.  

Thus by projecting the Sobolev gradient of $E$ at $z(t)$ into the nullspace of $\beta'(z(t))$ for each $t$, we obtain that $z(t)$ satisfies the constraint for all $t$ and in the limit $u$.  We will also show that if $u=\lim_{t \rightarrow \infty} z(t)$ for $z$ as given in equation \eqref{sd}, then $E'(u)h=0$ for all $h$ in the nullspace of $\beta'(u)$ which will give us the $u$ that we desired above.  We also provide a convergence result for $z$ in section \ref{steepest_descent}. In particular, if we assume that the trapping potential is bounded away from zero and non-negative we obtain that $\lim_{t \rightarrow \infty} z(t)$ exists.

In our formulation we will need to know what the projection $P_u$ is.  The next part of this work is devoted to finding an expression that we can use both in our analysis and in doing our simulations.

%%%%%%%%%%%%%%%%%%%%%%%%%%%%%%%%%%%%%%%%%%%%%%%%%%%%%
\subsection{Properties of the projection onto the nullspace of $\beta'(u)$}
\label{proj}
%%%%%%%%%%%%%%%%%%%%%%%%%%%%%%%%%%%%%%%%%%%%%%%%%%%%%

In this section we discuss the previously mentioned projection $P_u$ which constitutes an essential part of the Sobolev gradient. Before we show the explicit representation  for the projection $P_u$, we need to introduce the following definition and a theorem which is proved in \cite{pkjwn}.  
\vspace{3mm}
\begin{definition} \label{M}
Let $f \in L^2(D)$, and define $\alpha_f$ from $H^{1,2}(D)$ to $\R$ so that $\alpha_f(u)=\langle u , f \rangle_{L^2(D)}$.  Since $H^{1,2}(D)$ is continuously embedded in $L^2(D)$, $\alpha_f$ is continuous from $H^{1,2}(D)$ to $\R$.  Thus there exists $v \in H^{1,2}(D)$ so that $\langle u,f \rangle_{L^2(D)} = \langle u , v \rangle_{H^{1,2}(D)}$ for all $u \in H^{1,2}(D)$.  Define $Mf=v$.  
\end{definition}
\vspace{3mm}
\begin{theorem} \label{Mcont}
$M$ is injective.  $M \in L(X,Y)$ with $X=H^{1,2}(D)$ or $L^2(D)$ and $Y=H^{1,2}(D)$ or $L^2(D)$. In each case, the operator norm of $M$ is less than or equal to one.
\end{theorem}
 
These two results yield the following relationship between the two inner products $\langle \cdot, \cdot \rangle_{L^2(D)}$ and $\langle \cdot , \cdot \rangle_{H^{1,2}(D)}$
\BEW
\langle u,f \rangle_{L^2(D)} = \langle u , Mf \rangle_{H^{1,2}(D)}
\EEW
for all $u \in H^{1,2}(D)$ and $f \in L^2(D)$. 

We also have the following theorem for which we start with $f \in H^{1,2}(D)$ and let $Wf=\nabla f$. Thus, $W$ is a closed densely defined linear operator from $L^2(D)$ to $[L^2(D)]^d$ whose adjoint we denote  as $W^*$. The following representation for $M$ is proved in \cite{pkgl}.
\vspace{3mm}
\begin{theorem} \label{Midentity}
Let $M$ be the transformation defined in definition \ref{M}, then $M=(I+W^* W)^{-1}$.
\end{theorem}

\pagebreak 

Using the same formulation given in Definition \ref{M} for $L^2(D,\C)$ and $H^{1,2}(D,\C)$, one can also obtain an operator $M_{\C} : L^2(D,\C) \rightarrow H^{1,2}(D,\C)$ so that $\langle h , u \rangle_{L^2} = \langle h , M_{\C}u \rangle_H$ for all $h \in H$ and $u \in L^2$.  Then one has the following result.
\vspace{3mm}
\begin{theorem}
If $u=r+is \in L^2(D,\C)$, then $M_{\C}u=Mr + iMs$.
\end{theorem}
\begin{proof}
We have that $\langle h , u \rangle_{L^2} = \langle h , M_{\C} u \rangle_H$ for all $h \in H$.  First suppose that $h=f + 0i$ which yields 
\begin{equation*}
\Re \langle h , u \rangle_{L^2} = \Re \langle h , M_{\C} u \rangle_H = \langle f , \Re(M_{\C} u) \rangle_{H^{1,2}(D)} \, ,
\end{equation*}
where $\Re$ denotes the real part. Furthermore, we observe that
\begin{equation*}
\Re (\langle h , u \rangle_{L^2} )= \langle f , r \rangle_{L^2(D)} = \langle f , Mr \rangle_{H^{1,2}(D)}.
\end{equation*}  
and thus
\begin{equation*}
\langle f , Mr \rangle_{H^{1,2}(D)} =  \langle f , \Re(M_{\C} u) \rangle_{H^{1,2}(D)}
\end{equation*}
for all $f \in H^{1,2}(D)$ and it must be that $Mr = \Re(M_{\C} u)$.  A similar argument shows that $Ms= \Im (M_{\C} u)$, where $\Im$ denotes the imaginary part.
\end{proof}
From this result it follows that $M_{\C}$ has the same properties as $M$.  From now on we will not distinguish between $M$ and $M_{\C}$ in our notation.  For $u \in L^2(D,\C)$, $M u=Mr + iMs$.
%\pagebreak

Using these definitions and results, we obtain an explicit formula for the projection $P_u$
\BE \label{eqnpu}
P_u h=(I - Q_u^*(Q_uQ_u^*)^{-1}Q_u) h= h - \frac{\Re \langle u , h \rangle_{L^2}}{\Re \langle u , Mu \rangle_{L^2}} Mu \, ,
\EE
which is derived in Appendix \ref{projection}.

Looking at this expression for the projection $P_u$, we note that as $u$ varies, the associated projections vary in a continuous way. We will use this result to obtain global existence for $z$.  In more precise terms, we have the following result which is proved in Appendix \ref{lipschitz}.
\vspace{3mm}
\begin{proposition} \label{pu}
Suppose $\{u_n\}_{n \geq 1}$ is a sequence of members of $H$ that converges in $L^2$ to $u \neq 0 \in H$.  Then the sequence $\{P_{u_n} \}_{n \geq 1}$ converges in $L(H,H)$ to $P_u$.  Furthermore there is a constant $m$ so that $\|P_{u_n} - P_u \| \leq m\|u_n - u \|_{L^2}$.
\end{proposition}

\newpage

%%%%%%%%%%%%%%%%%%%%%%%%%%%%%%%%%%%%%%%%%%%%%%%%%%%%%
\subsection{Global existence, uniqueness and convergence of the Sobolev gradient}
\label{steepest_descent}
%%%%%%%%%%%%%%%%%%%%%%%%%%%%%%%%%%%%%%%%%%%%%%%%%%%%%

In this section we first show that $z$ as given in equation \eqref{sd} is defined for all $ t \geq 0$ by using Proposition \ref{pu}. Additionally, we prove that $\lim_{t \rightarrow \infty} z(t)$ exists in $H$ if the trapping potential is bounded away from zero by some positive number $\epsilon$. 
\vspace{3mm}
\begin{theorem} \label{thmsdexists}
$z$ as given in \eqref{sd} is uniquely defined for all $t \geq 0$.
\end{theorem}

\begin{proof}
Since $E$ is continuously twice differentiable, the map $u \rightarrow \nabla_H E(u)$ from $H$ to $H$ is a Lipschitz map.  Furthermore, due to the result obtained in Proposition \ref{pu}, we have that the map $u \rightarrow P_u \nabla_H E(u)$ is also a Lipschitz map.  This implies that we have local existence for equation \eqref{sd}.  Suppose that there exists a number $T$ so that $z$ as defined in equation \eqref{sd} can only be defined for $0 \leq t <T$.  To see that we have global existence, let $0 \leq a < b < T$ and note that 
\begin{eqnarray*}
\|z(b) - z(a) \|_H^2 & \leq & \left(\int_a^b \|z'(t)\|_H \,\d t\right)^2 \\
&=& \left(\int_a^b \|P_{z(t)} \nabla_H E(z(t))\|_H\,\d t\right)^2  \\
&\leq& (b-a) \int_a^b \|P_{z(t)} \nabla_H E(z(t))\|_H^2\, \d t \\
&\leq& T \int_a^b \|P_{z(t)} \nabla_H E(z(t))\|_H^2\, \d t \, .
\end{eqnarray*}  

Furthermore, we have 
\begin{eqnarray*}
\big(E(z)\big)'(t)&=& E'(z(t))z'(t)\\
&=& \langle z'(t) , \nabla_H E(z(t)) \rangle_H \\
&=&-\langle P_{z(t)} \nabla_H E(z(t)), \nabla_H E(z(t)) \rangle_H  \\
&=&-\langle P_{z(t)}^2 \nabla_H E(z(t)), \nabla_H E(z(t)) \rangle_H  \\
&=&-\langle P_{z(t)} \nabla_H E(z(t)), P_{z(t)} \nabla_H E(z(t)) \rangle_H \\
&=&-\| P_{z(t)} \nabla_H E(z(t))\|_H^2 \leq 0 
\end{eqnarray*}
for all $t < T$ which implies that $E(z)$ is decreasing on $[0,T)$. Additionally, $\int_0^T \| P_{z(t)} \nabla_H E(z(t))\|_H^2 \, \d t < \infty $ since
\BEW
\int_a^b \|P_{z(t)} \nabla_H E(z(t))\|_H^2 \, \d t %&=& - \int_a^b (E(z))'(t) dt \\
&=& E(z(a)) - E(z(b)) \leq E(z(0)) \,.
\EEW

\pagebreak 

Thus $\lim_{t \rightarrow T^-} z(t) = z(T)$ exists and one can therefore define $z$ on $[0,T+ \epsilon)$ for some positive number $\epsilon$, which means that $z(t)$ must exist for all $t \geq 0$. Uniqueness follows from basic existence and uniqueness for ordinary differential equations.
\end{proof}

This global existence and uniqueness result for $z$ as given in equation \eqref{sd} is accompanied by global convergence to an element of the Sobolev space $u \in H$. Our result is an adaptation of what is presented in \cite{jwn} and implies strong convergence in the $H$ norm. This is a stronger result than that of using a minimizing sequence to obtain weak $H$ convergence or strong $L^p$ convergence as has been previously done. For our proof, suppose that the trapping potential $V_{trap}$ is bounded away from zero by some positive number $\epsilon$.  This does not lead to any loss of generality as previously described. Recall that this assumption implies that there is a positive number $\delta$ so that 
\BEW
E''(u)(h,h) \geq \delta \|h\|_H^2 \, .
\EEW

%\pagebreak

\begin{theorem}
Suppose that $z$ is given by equation \eqref{sd}, with $\nabla_H E(u_0) \neq 0$, then
\BEW
\lim_{ t \rightarrow \infty} z(t) = u
\EEW
exists. Furthermore, there exist constants $m$ and $c$ so that $\|u - z(t)\|_H \leq m e^{-ct}$, and $E'(u)h=0$ for all $h \in null (\beta'(u))$.  

\end{theorem}

\begin{proof}
Let $g(t)=E(z(t))$ then we know already that 
\begin{eqnarray*}
g'(t) & = & - \| P_{z(t)} \nabla_H E(z(t)) \|_H^2 \, .
\end{eqnarray*}
Moreover, let $G: H \rightarrow H$ be defined by $G(x)=P_x \nabla_H E(x)$ then $G$ is $C^1$ and thus
\BEW
g''(t)=  -2 \langle G'(z(t))z'(t), G(z(t)) \rangle_H \, . 
\EEW
Now note that if $h \in null (\beta'(z(t))$, then 
\begin{eqnarray*}
E'(z(t))(h) &=& \langle h , \nabla_H E(z(t)) \rangle_H = \langle P_{z(t)} h , \nabla_H E(z(t)) \rangle_H = \langle h , G(z(t)) \rangle_H 
\end{eqnarray*}
and therefore
\BEW
E''(z(t))(h,z'(t)) = \langle h , G'(z(t))z'(t) \rangle_H \, .
\EEW
Since 
\BEW
z'(t)= -P_{z(t)} \nabla_H E(z(t)) = - G(z(t))
\EEW
is in the nullspace of $\beta'(z(t))$, one has
\begin{eqnarray*}
E''(z(t))(z'(t),z'(t)) &=& \langle z'(t) , G'(z(t))z'(t) \rangle_H \\
& = & \langle -G(z(t)) , G'(z(t))z'(t) \rangle_H =  \frac{g''(t)}{2} \, .
\end{eqnarray*}
Making use of Lemma \ref{convexity} we consequently have
\BEW
\frac{g''(t)}{2} &\geq & \delta \|z'(t)\|_H^2 = \delta \|-G(z(t))\|_H^2 = \delta \| G(z(t)) \|_H^2 = - \delta g'(t)
\EEW
and hence
\BEW
- \frac{g''(t)}{g'(t)} \geq 2 \delta \, .
\EEW
Integrating both sides from 0 to $t$, we get that
\BEW
-\ln(-g'(t)) + \ln(-g'(0)) \geq 2 \delta t
\EEW
and thus there is a constant $m$ so that 
\BEW
0 \leq - g'(t) \leq m e^{-2\delta t} \quad \text{ for all $t$.}
\EEW
More specifically $m= \|P_{z(0)} \nabla_H E(z(0))\|_H^2$. From this we see that
\BEW
\left(\int_n^{n+1} \|z'\|_H \,\d t\right)^2 &\leq& \int_n^{n+1} \|z'\|_H^2 \,\d t = - \int_n^{n+1}  g'\,\d t \leq  m \int_n^{n+1}  e^{-2\delta t} \d t  
\EEW
and therefore
\BEW
\int_0^{\infty} \|z'\|_H \,\d t < \infty
\EEW
which implies that $\lim_{t \rightarrow \infty} z(t)=u$ exists. The rate of convergence is given by 
\BEW
\|u - z(n)\|_H \leq \left( \frac{\sqrt{m}}{\sqrt{2 \delta}(1 - e^{- \delta})} \right) e^{- \delta n} \, .
\EEW
 
To see that $E'(u)h=0$ for all $h \in null(\beta'(u))$, recall that by Proposition \ref{pu}, the map $u \rightarrow P_u$ from $H$ to $L(H,H)$ is Lipschitz. 

Furthermore since $E$ is $C^2$, the map $ u \rightarrow \nabla_H E(u)$ is Lipschitz also from $H$ to $H$. Thus the map $u \rightarrow P_u \nabla_H E(u)$ is Lipschitz and from this it follows that
\BEW
P_u \nabla_H E(u) = \lim_{ t \rightarrow \infty} P_{z(t)} \nabla_{H} E(z(t)).
\EEW
Since
\BEW
\int_0^{\infty} \|z'\|_H \,\d t = \int_0^{\infty} \|P_{z(t)} \nabla_{H} E(z(t))\|_H \,\d t
\EEW
is finite, then it must be that $P_u \nabla_H E(u)=0$.  Now let $h \in null(\beta'(u))$ then
\BEW
E'(u)h= \langle h , \nabla_H E(u) \rangle_H = \langle P_u h , \nabla_H E(u) \rangle_H = \langle h , P_u \nabla_H E(u) \rangle_H = 0 \, ,
\EEW
which concludes the proof.

\end{proof}

\newpage

%%%%%%%%%%%%%%%%%%%%%%%%%%%%%%%%%%%%%%%%%%%%%%%%%%%%%
\section{Numerical results} 
%%%%%%%%%%%%%%%%%%%%%%%%%%%%%%%%%%%%%%%%%%%%%%%%%%%%%

In the second part of our contribution we want to apply the results that we obtained in the detailed discussion about the theoretical background of the Sobolev method and its application to calculate the ground state of the Gross-Pitaevskii energy functional. In particular we want to demonstrate the extraordinary convergence properties of this method and highlight its numerical efficiency. For this purpose we do not use the commonly known harmonic oscillator, but the more challenging generalized Mexican hat potential as our trapping potential which in natural units reads
\BEW
V_{trap}(\x) & = & A\left( \sum\limits_{i=1}^d (C_i x_i)^2 - B \right)^2
\EEW
for the $d$-dimensional case. For our numerical simulations we choose $A=\frac{1}{10}$, $B=16$ and -- to break the spherical symmetry -- $C_1 = 1.0$, $C_2 = 1.5$ and $C_3=2.0$.

Furthermore, we vary the coupling constant $g$ from $10^{-1}$ over $10^{0}$ up to $10^{1}$ such that we cover the regime from a nearly interactionless gas of bosons to one which is interaction dominated. For $g=10^{-1}$ the bosons are only weakly interacting, whereas they are strongly interacting for $g=10^{1}$. For a numerical realization of the minimization we also need the particle number $N$ which determines the normalization constraint. The particle numbers we choose for our illustration purposes are $N=10^2$ for the one-dimensional case, $N=10^3$ in two dimensions and $N=10^4$ in the three-dimensional case.

Before we discuss our results, a few words are in place about our numerical implementation of the method of Sobolev gradients. The details about our implementation are summarized in Appendix C and we briefly want to mention that we use a discretization in position space and an implementation of the differentiation by means of central differencing \cite{Renka_JWN}. Furthermore, we use an Euler iteration to discretize equation \eqref{sd} and optimize the step size for each iteration. Once the relative change in energy from one iteration step to the next is small enough -- in our case this means less than $10^{-4}$ for 1D, $10^{-3}$ for 2D and $10^{-4}$ for 3D due to the specific choice of the potential and the respective particle numbers -- we switch to the well-known Newton method in order to find the exact solution of the stationary Gross-Pitaevskii equation
\BEW
- \frac{1}{2} \nabla^2 u(\x) + V_{trap}(\x) u(\x) + g |u(\x)|^2u(\x) = \mu u(\x) \, ,
\EEW
where $\mu$ denotes the so-called chemical potential. 

The ground state of the Gross-Pitaevskii energy functional has to fulfill this equation and only few steps with the Newton method are necessary to arrive at a desired accuracy of $10^{-8}$ in each component. 

To clearly demonstrate the superior convergence properties of our numerical simulations with the Sobolev gradient, we choose the initial wave function to be a real-valued random field with zero boundary conditions which on average is as far as possible from the desired solution. All simulations were run on a \textit{Dell PowerEdge 2900 Server} with four Intel Xeon 5160 CPUs at a frequency of 3.00 GHz and a total of 16 GB of RAM. The MATLAB program for our simulations is a 64-bit Linux version 7.5.0.338 (R2007b).

\subsection{Simulations in 1D}

\begin{table}[h]
\label{comp_1d}

\begin{center}
\begin{tabular}{C{0.6cm}|C{0.8cm}||C{1.1cm}|C{1.15cm}|C{1.15cm}|C{1.15cm}|C{1.15cm}|C{1.4cm}|C{1.4cm}}
\hline\hline 
$\bm{N_x}$ & $\bm{g}$ & $\bm{\mu}$ & $\bm{\#_{\text{S,Min}}}$ & $\bm{\#_{\text{S,Max}}}$ & $\bm{\#_{\text{N,Min}}}$ & $\bm{\#_{\text{N,Max}}}$ & $\bm{t_{\text{Min}}}$ & $\bm{t_{\text{Max}}}$ \\ \hline\hline 
&&&&&&&\\[-6mm]
& $10^{-1}$ & $4.8938$ & $130$ & $241$ & $17$ & $18$ & $5.2 \cdot 10^{-1}$ & $8.6 \cdot 10^{-1}$ \\\cline{2-9}
&&&&&&&\\[-6mm]
$2^{7\phantom{0}}$ & $10^{0\phantom{-}}$ & $19.831$ & $98$ & $123$ & $21$ & $23$ & $4.4 \cdot 10^{-1}$ & $5.9 \cdot 10^{-1}$ \\\cline{2-9}
&&&&&&&\\[-6mm]
& $10^{1\phantom{-}}$ & $92.607$ & $54$ & $86$ & $19$ & $23$ & $4.2 \cdot 10^{-1}$ & $5.7 \cdot 10^{-1}$\\\hline\hline
&&&&&&&\\[-6mm]
& $10^{-1}$ & $4.8952$ & $156$ & $278$ & $18$ & $19$ & $6.6 \cdot 10^{-1}$ & $1.1 \cdot 10^{0\phantom{-}}$ \\\cline{2-9}
&&&&&&&\\[-6mm]
$2^{8\phantom{0}}$ & $10^{0\phantom{-}}$ & $19.831$ & $109$ & $179$ & $20$ & $23$ & $5.0 \cdot 10^{-1}$ & $7.4 \cdot 10^{-1}$ \\\cline{2-9}
&&&&&&&\\[-6mm]
& $10^{1\phantom{-}}$ & $92.607$ & $80$ & $121$ & $20$ & $20$ & $4.1 \cdot 10^{-1}$ & $5.5 \cdot 10^{-1}$\\\hline\hline
&&&&&&&\\[-6mm]
& $10^{-1}$ & $4.8956$ & $153$ & $309$ & $18$ & $20$ & $7.6 \cdot 10^{-1}$ & $1.4 \cdot 10^{0\phantom{-}}$ \\\cline{2-9}
&&&&&&&\\[-6mm]
$2^{9\phantom{0}}$ & $10^{0\phantom{-}}$ & $19.832$ & $124$ & $208$ & $20$ & $23$ & $6.5 \cdot 10^{-1}$ & $1.0 \cdot 10^{0\phantom{-}}$ \\\cline{2-9}
&&&&&&&\\[-6mm]
& $10^{1\phantom{-}}$ & $92.608$ & $93$ & $147$ & $21$ & $21$ & $5.3 \cdot 10^{-1}$ & $7.6 \cdot 10^{-1}$ \\\hline\hline
\end{tabular}
\end{center}
\vspace{2mm}
\caption{Simulation results with the {Sobolev} gradient for the ground state of the 1D {Gross-Pitaevskii} energy functional with a Mexican hat potential.}
\end{table}

In table \ref{comp_1d} we see the results of the one-dimensional simulations for a particle number of $N=10^2$ and a grid length of $L_x=10$. For various discretization numbers $N_x$ and coupling strengths $g$ we run our simulations to obtain the ground state for 10 different random initial fields. We show the respective chemical potential $\mu$ in dimensionless units, the minimal and maximal number of iteration steps with the Sobolev gradient $\#_{\text{S,Min/Max}}$, the minimal and maximal number of iteration steps with the Newton method $\#_{\text{N,Min/Max}}$ and the minimal and maximal CPU time $t_{\text{Min/Max}}$ for the simulation in seconds, which was obtained in our MATLAB simulation with the built-in functions \textit{tic} and \textit{toc}. 

We can clearly observe the dependence of the number of iteration steps with the Sobolev gradient on the different random initial fields. In general the number of necessary steps decreases for an increasing coupling constant. In contrast, the number of iterations with the Newton method is almost independent of the random initial field and shows a negligible dependence on the coupling constant. Furthermore, it is worth mentioning that the simulation time is independent of the number of grid points since the complexity of the minimization in one dimension is fairly low and thus the computational overhead is dominant.

The shape of the density distribution $n_x = |u(x)|^2$ is depicted in figure \ref{wave_1D_comp} for $N_x=2^{9}$ and interaction strengths of $g=10^{-1}$, $10^{0}$ and $10^1$. We can notice the change from a clear double Gaussian shape for a small interaction strength towards a strongly interaction broadened shape at a large interaction strength.

\begin{figure}
\centering
\includegraphics[width=0.29\columnwidth]{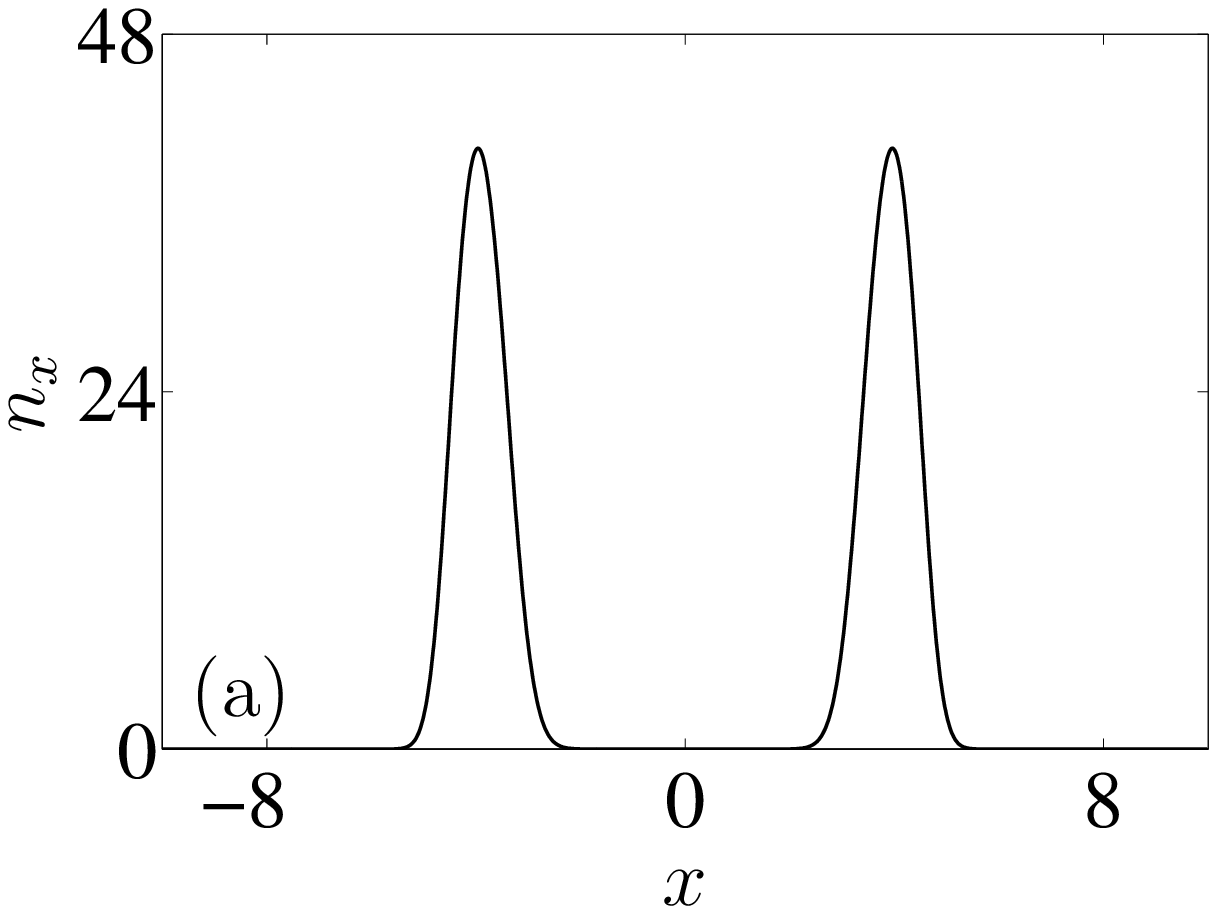}
\includegraphics[width=0.29\columnwidth]{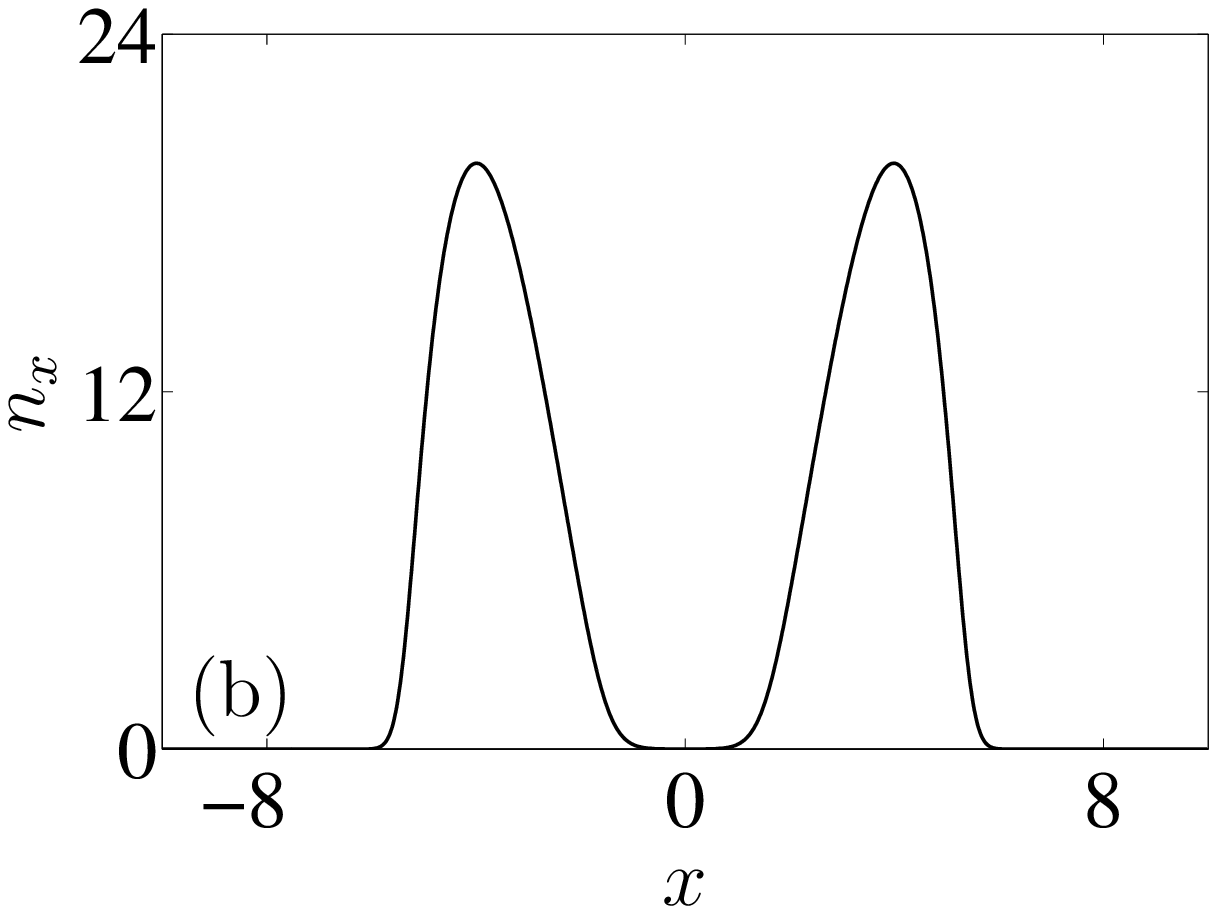}
\includegraphics[width=0.29\columnwidth]{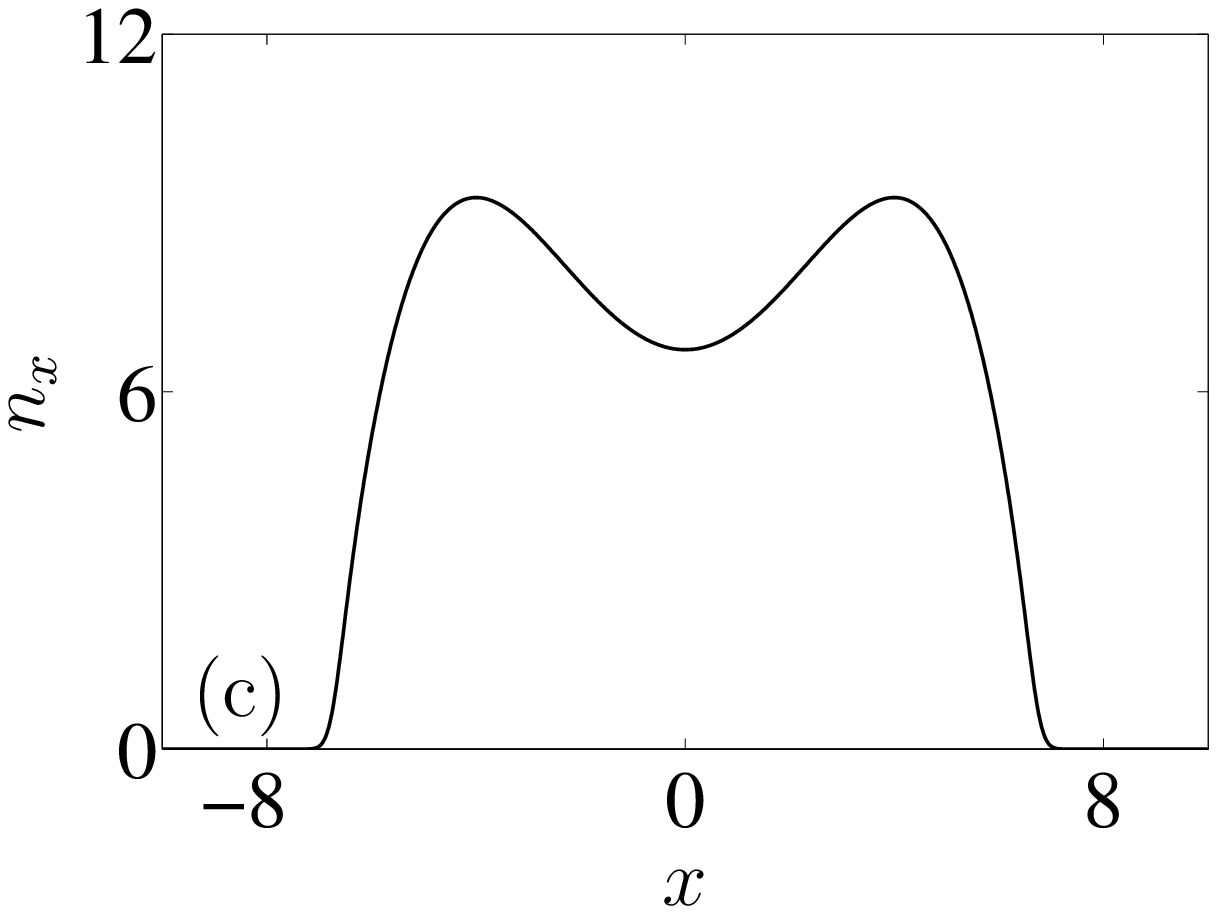}
\caption{Density distribution $n_x = |u(x)|^2$ of the ground state of the 1D {Gross-Pitaevskii} energy functional with a Mexican hat potential for a particle number of $N=10^2$. The number of grid points for the simulation is $N_x=2^{9}$ and the grid size is $L_x=10$. The interaction strength is $g=10^{-1}$ in subplot (a), $g=10^{0}$ in subplot (b) and $g=10^{1}$ in subplot (c).}
\label{wave_1D_comp}
\end{figure}

\subsection{Simulations in 2D}

In table \ref{comp_2d} we compare the parameters for the two-dimensional simulation for a particle number of $N=10^3$ and a grid length of $L_x=L_y=10$. The discretization for both spatial parameters $x$ and $y$ was chosen to be equal by setting $N_x=N_y$. As previously, we show the respective chemical potential $\mu$ in dimensionless units, the minimal and maximal number of iteration steps with the Sobolev gradient $\#_{\text{S,Min/Max}}$, the minimal and maximal number of iteration steps with the Newton method $\#_{\text{N,Min/Max}}$ and the minimal and maximal CPU time $t_{\text{Min/Max}}$ for the simulation in seconds.

%\pagebreak

\begin{table}[h]
\label{comp_2d}

\begin{center}
\begin{tabular}{C{0.6cm}|C{0.8cm}||C{1.1cm}|C{1.15cm}|C{1.15cm}|C{1.15cm}|C{1.15cm}|C{1.4cm}|C{1.4cm}}
\hline\hline 
$\bm{N_x}$ & $\bm{g}$ & $\bm{\mu}$ & $\bm{\#_{\text{S,Min}}}$ & $\bm{\#_{\text{S,Max}}}$ & $\bm{\#_{\text{N,Min}}}$ & $\bm{\#_{\text{N,Max}}}$ & $\bm{t_{\text{Min}}}$ & $\bm{t_{\text{Max}}}$ \\ \hline\hline 
&&&&&&&\\[-6mm]
& $10^{-1}$ & $5.6656$ & $350$ & $374$ & $21$ & $23$ & $5.7 \cdot 10^{0}$ & $9.9 \cdot 10^{0}$ \\\cline{2-9}
&&&&&&&\\[-6mm]
$2^{6}$ & $10^{0\phantom{-}}$ & $23.557$ & $179$ & $190$ & $23$ & $26$ & $5.1 \cdot 10^{0}$ & $5.4 \cdot 10^{0}$ \\\cline{2-9}
&&&&&&&\\[-6mm]
& $10^{1\phantom{-}}$ & $124.44$ & $87$ & $103$ & $25$ & $27$ & $4.4 \cdot 10^{0}$ & $5.3 \cdot 10^{0}$\\\hline\hline
&&&&&&&\\[-6mm]
& $10^{-1}$ & $5.6836$ & $427$ & $451$ & $20$ & $21$ & $1.1 \cdot 10^{2}$ & $1.5 \cdot 10^{2}$ \\\cline{2-9}
&&&&&&&\\[-6mm]
$2^{7}$ & $10^{0\phantom{-}}$ & $23.564$ & $152$ & $243$ & $23$ & $23$ & $4.6 \cdot 10^{1}$ & $6.1 \cdot 10^{1}$ \\\cline{2-9}
&&&&&&&\\[-6mm]
& $10^{1\phantom{-}}$ & $124.44$ & $93$ & $176$ & $25$ & $30$ & $3.5 \cdot 10^{1}$ & $6.3 \cdot 10^{1}$ \\\hline\hline
&&&&&&&\\[-6mm]
& $10^{-1}$ & $5.6883$ & $615$ & $990$ & $24$ & $24$ & $9.2 \cdot 10^{2}$ & $2.3 \cdot 10^{3}$ \\\cline{2-9}
&&&&&&&\\[-6mm]
$2^{8}$ & $10^{0\phantom{-}}$ & $23.566$ & $324$ & $626$ & $22$ & $24$ & $6.0 \cdot 10^{2}$ & $1.1 \cdot 10^{3}$ \\\cline{2-9}
&&&&&&&\\[-6mm]
& $10^{1\phantom{-}}$ & $124.44$ & $142$ & $236$ & $26$ & $28$ & $4.6 \cdot 10^{2}$ & $6.7 \cdot 10^{2}$\\
\hline\hline
\end{tabular}
\end{center}
\vspace{2mm}
\caption{Simulation results with the {Sobolev} gradient for the ground state of the 2D {Gross-Pitaevskii} energy functional with a Mexican hat potential.}
\end{table}

For the two-dimensional case the number of iteration steps with the Sobolev gradient also depends on the different random initial fields as was the case for the one-dimensional simulations. The number of iteration steps with the Newton method has again a negligible dependence on the different random initial fields. Overall, the number of necessary steps has a strong dependence on the coupling constant and is significantly lower for a large coupling constant, which is due to a much smoother behavior of the respective ground state wave function. Now the simulation time depends on the number of grid points since the complexity of the minimization in two dimensions is rapidly growing. Combining the increasing number of iteration steps for twice as many grid points, the simulation time approximately grows by a factor of four which is exactly what one would expect since the complexity of a two dimensional system also grows by this factor when the number of grid points is doubled. Nevertheless, the simulation time is on the order of several minutes despite a maximum system size of $2^{16}$ grid points. This clearly demonstrates the efficiency of the method of Sobolev gradients.

To illustrate the shape of the density distribution $n_\x = |\alpha_\x|^2$, we depict in figure \ref{wave_2D_comp} cuts along the $x$-axis for $y=0$ (black lines) and along the $y$-axis for $x=0$ (blue lines). For these plots we used a discretization number of $N_x=2^{8}$ and interaction strengths of $g=10^{-1}$, $10^{0}$ and $10^1$. Once more we can notice the change from a clear double Gaussian shape for a small interaction strength towards a strongly interaction broadened shape at a large interaction strength. Since we deal with an anisotropic trapping potential the width in the $y$-direction is smaller than in the $x$-direction.

\begin{figure}
\centering
\includegraphics[width=0.30\columnwidth]{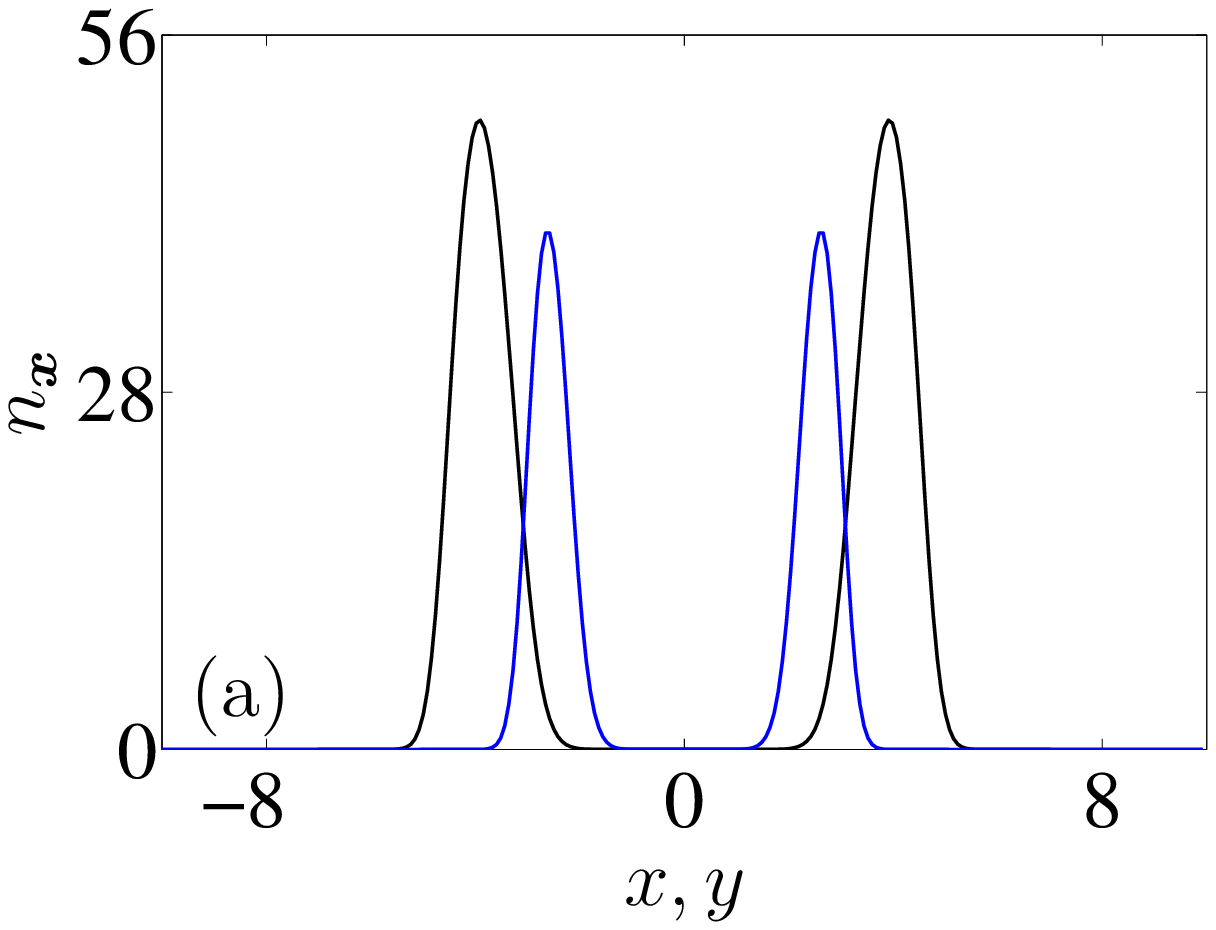}
\includegraphics[width=0.30\columnwidth]{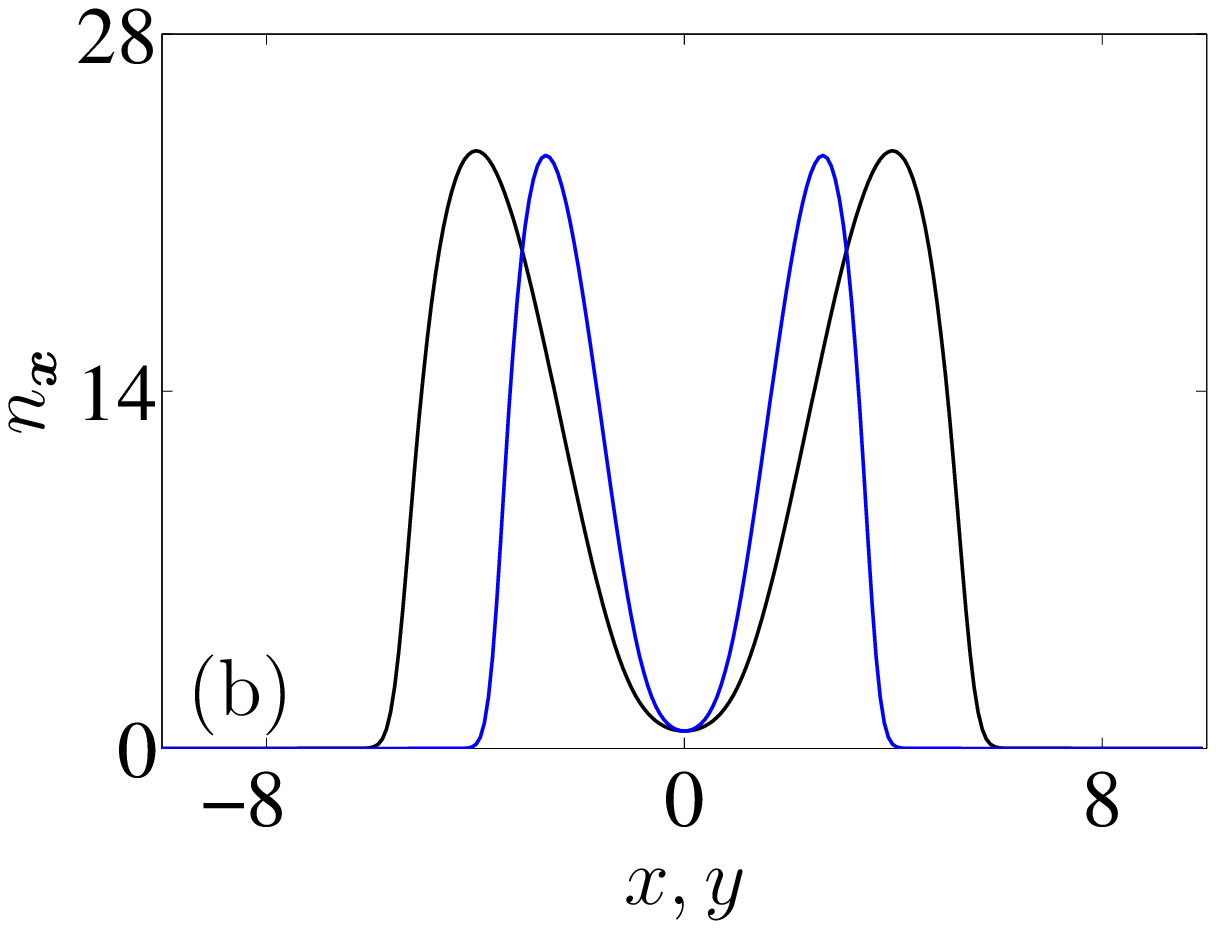}
\includegraphics[width=0.30\columnwidth]{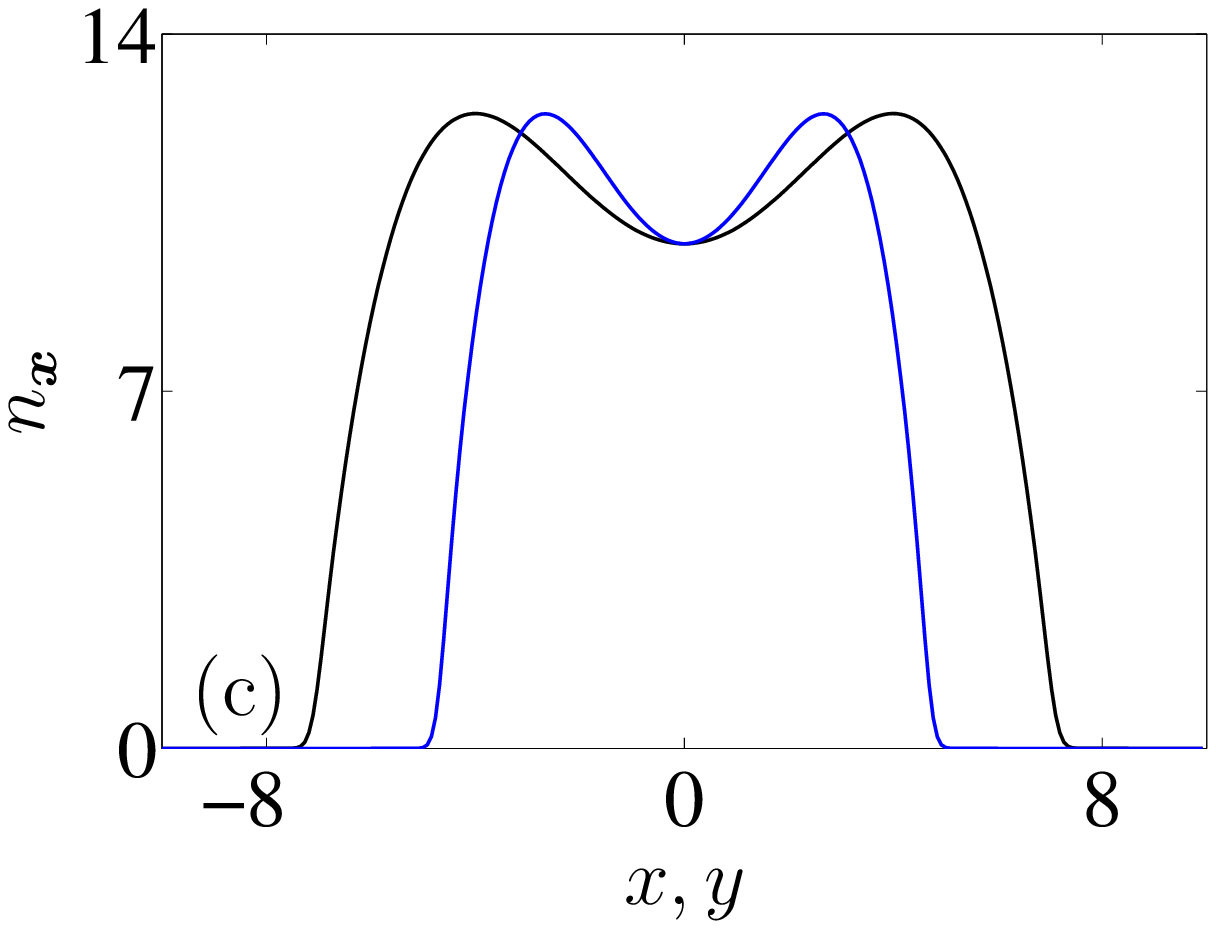}
\caption{Cuts through the density distribution $n_\x = |u(\x)|^2$ of the ground state of the 2D {Gross-Pitaevskii} energy functional with a Mexican hat potential. The cuts are along the $x$-axis for $y=0$ (black lines) and along the $y$-axis for $x=0$ (blue lines). The results are for a particle number of $N=10^3$, the number of grid points for the simulation is $N_x=N_y=2^{8}$ and the grid size is $L_x=L_y=10$. The interaction strength is $g=10^{-1}$ in subplot (a), $g=10^{0}$ in subplot (b) and $g=10^{1}$ in subplot (c).}
\label{wave_2D_comp}
\end{figure}

\subsection{Simulations in 3D}

\begin{table}[h]
\label{comp_3d}

\begin{center}
\begin{tabular}{C{0.6cm}|C{0.8cm}||C{1.1cm}|C{1.15cm}|C{1.15cm}|C{1.15cm}|C{1.15cm}|C{1.4cm}|C{1.4cm}}
\hline\hline 
$\bm{N_x}$ & $\bm{g}$ & $\bm{\mu}$ & $\bm{\#_{\text{S,Min}}}$ & $\bm{\#_{\text{S,Max}}}$ & $\bm{\#_{\text{N,Min}}}$ & $\bm{\#_{\text{N,Max}}}$ & $\bm{t_{\text{Min}}}$ & $\bm{t_{\text{Max}}}$ \\ \hline\hline 
&&&&&&&\\[-6mm]
& $10^{-1}$ & $9.8716$ & $767$ & $955$ & $24$ & $28$ & $1.2 \cdot 10^{3}$ & $2.5 \cdot 10^{3}$ \\\cline{2-9}
&&&&&&&\\[-6mm]
$2^{5}$ & $10^{0\phantom{-}}$ & $45.302$ & $297$ & $341$ & $23$ & $28$ & $5.1 \cdot 10^{2}$ & $8.2 \cdot 10^{2}$ \\\cline{2-9}
&&&&&&&\\[-6mm]
& $10^{1\phantom{-}}$ & $218.23$ & $133$ & $152$ & $25$ & $29$ & $2.5 \cdot 10^{2}$ & $4.8 \cdot 10^{2}$ \\\hline\hline
&&&&&&&\\[-6mm]
& $10^{-1}$ & $9.9784$ & $755$ & $987$ & $24$ & $29$ & $1.7 \cdot 10^{4}$ & $2.4 \cdot 10^{4}$ \\\cline{2-9}
&&&&&&&\\[-6mm]
$2^{6}$ & $10^{0\phantom{-}}$ & $45.414$ & $354$ & $421$ & $22$ & $23$ & $7.8 \cdot 10^{3}$ & $1.1 \cdot 10^{4}$ \\\cline{2-9}
&&&&&&&\\[-6mm]
& $10^{1\phantom{-}}$ & $218.18$ & $199$ & $215$ & $25$ & $29$ & $2.9 \cdot 10^{3}$ & $5.7 \cdot 10^{3}$ \\\hline\hline
&&&&&&&\\[-6mm]
& $10^{-1}$ & $10.011$ & $1078$ & $1485$ & $23$ & $25$ & $1.5 \cdot 10^{5}$ & $2.3 \cdot 10^{5}$ \\\cline{2-9}
&&&&&&&\\[-6mm]
$2^{7}$ & $10^{0\phantom{-}}$ & $45.443$ & $640$ & $769$ & $23$ & $23$ & $7.8 \cdot 10^{4}$ & $1.4 \cdot 10^{5}$ \\\cline{2-9}
&&&&&&&\\[-6mm]
& $10^{1\phantom{-}}$ & $218.21$ & $371$ & $466$ & $26$ & $33$ & $5.0 \cdot 10^{4}$ & $6.9 \cdot 10^{4}$ \\
\hline\hline
\end{tabular}
\end{center}
\vspace{2mm}
\caption{Simulation results with the {Sobolev} gradient for the ground state of the 3D {Gross-Pitaevskii} energy functional with a Mexican hat potential.}
\end{table}

%\pagebreak

Eventually we compare the parameters for the three-dimensional simulation in table \ref{comp_3d}. We used a particle number of $N= 10^4$ and a grid length of $L_x=L_y=L_z=10$. The discretization for all spatial parameters is the same and thus $N_x=N_y=N_z$. Again we show the respective chemical potential $\mu$ in dimensionless units, the minimal and maximal number of iteration steps with the Sobolev gradient $\#_{\text{S,Min/Max}}$ and with the Newton method $\#_{\text{N,Min/Max}}$ and the minimal and maximal CPU time $t_{\text{Min/Max}}$ for the simulation in seconds.

As before the number of iteration steps with the Sobolev gradient depends on the different random initial fields whereas they do not influence the number of iteration steps with the Newton method. The number of necessary steps also strongly depends on the coupling constant and is lower for a large coupling constant, where the respective ground state wave function has a smoother behavior. The simulation time is highly dependent on the number of grid points due to the complexity of the three dimensional simulations. The simulation time is now on the order of one day for a maximum system size of $2^{21}$ grid points. However we want to point out, that this is not a standard computation time for a given potential but represents the upper limit for a ground state simulation with our method, since we started as far as possible from the final solution and with the most unfavorable initial state.

\begin{figure}
\centering
\includegraphics[width=0.30\columnwidth]{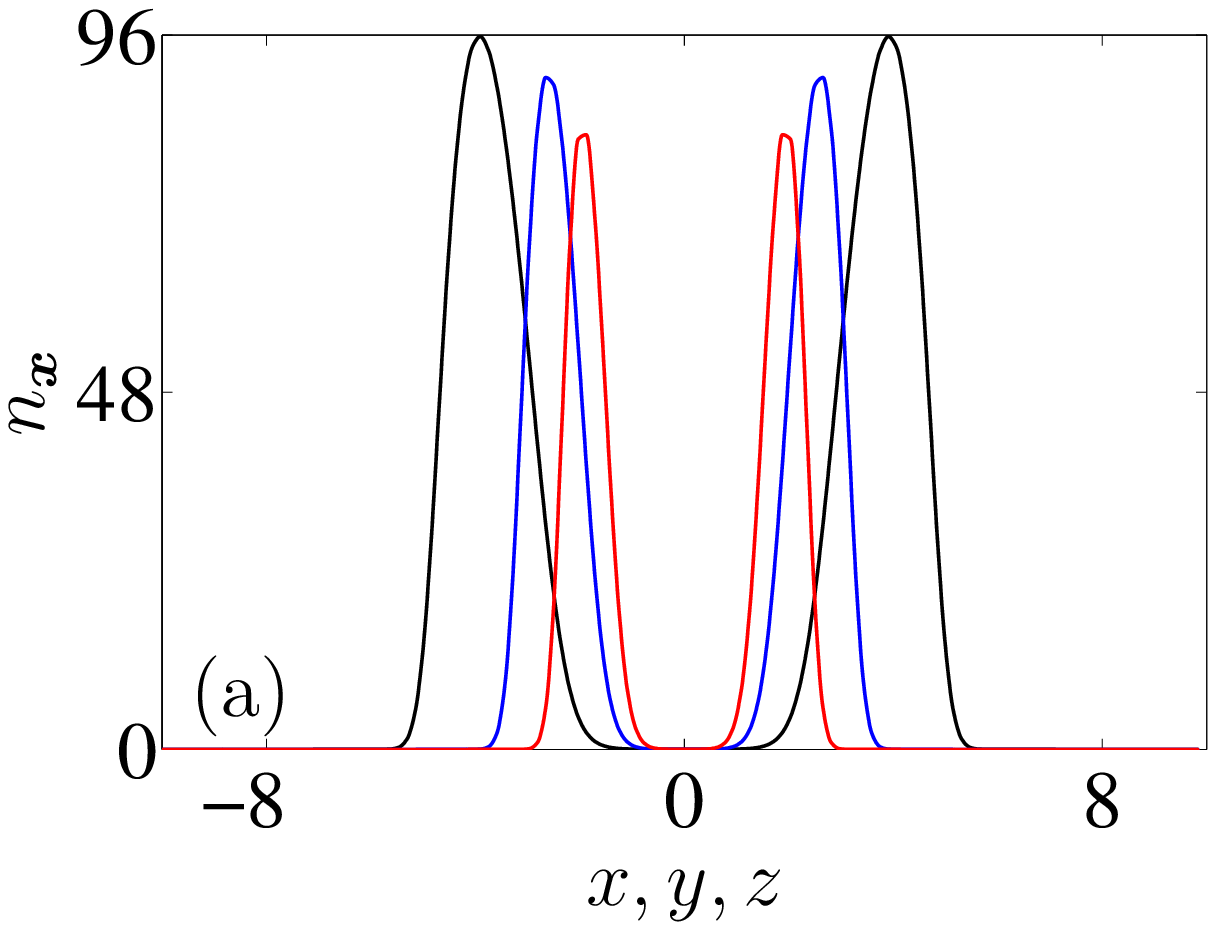}
\includegraphics[width=0.30\columnwidth]{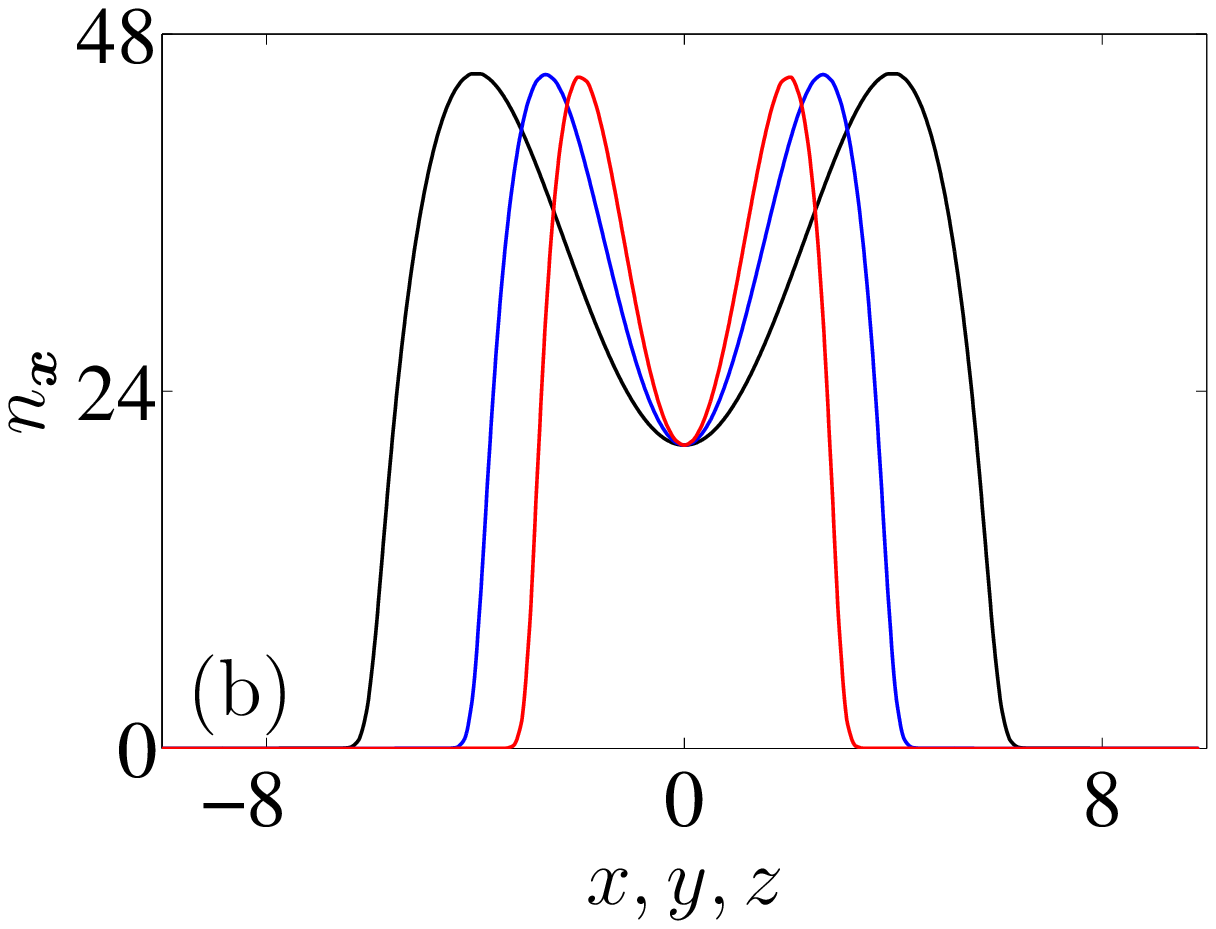}
\includegraphics[width=0.30\columnwidth]{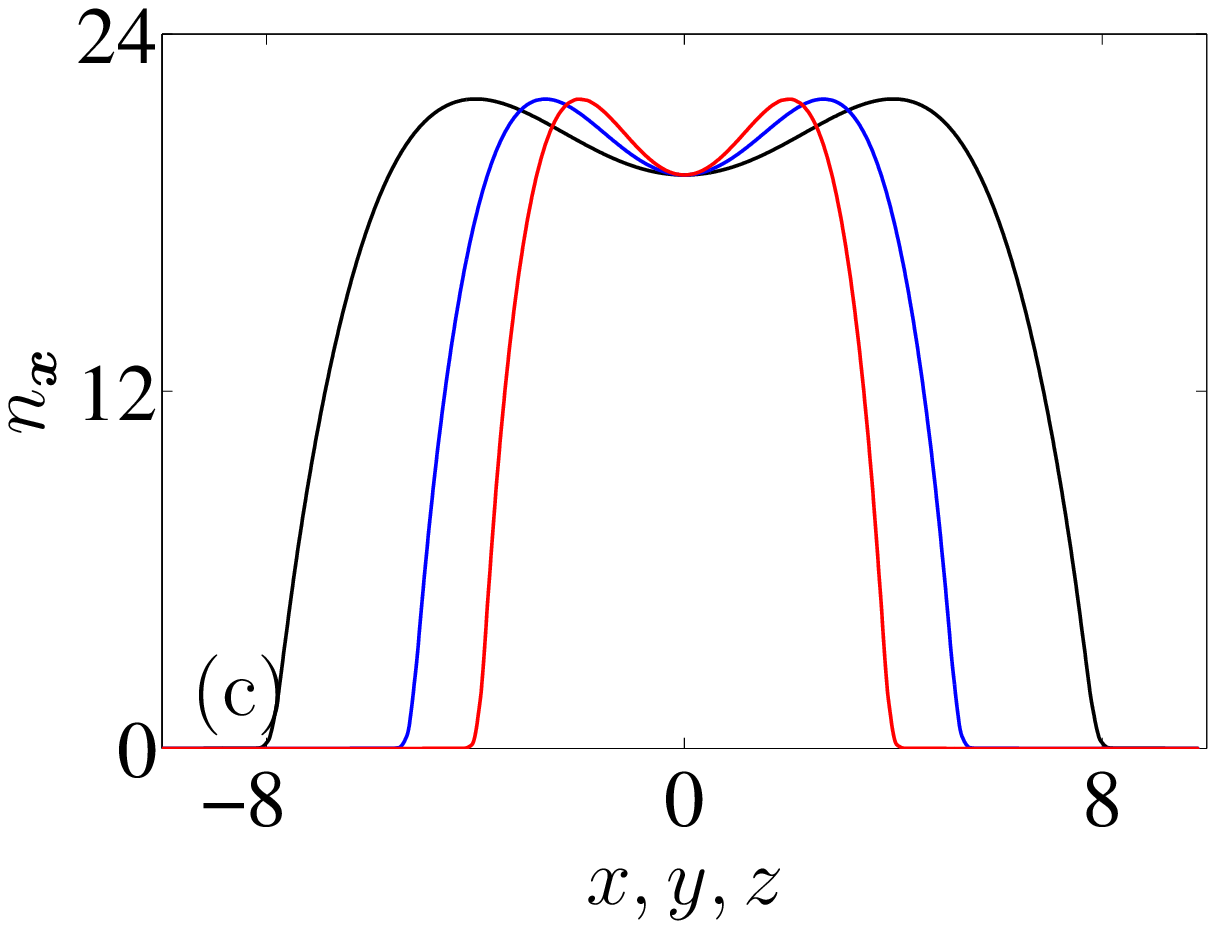}
\caption{Cuts through the density distribution $n_\x = |u(\x)|^2$ of the ground state of the 3D {Gross-Pitaevskii} equation with a Mexican hat potential. The cuts are along the $x$-axis for $y,z=0$ (black lines), along the $y$-axis for $x,z=0$ (blue lines) and along the $z$-axis for $x,y=0$ (red lines). The results are for a particle number of $N=10^4$,  the number of grid points for the simulation is $N_x=N_y=N_z=2^{7}$ and the grid size is $L_x=L_y=L_z=10$. The interaction strength is $g=10^{-1}$ in subplot (a), $g=10^{0}$ in subplot (b) and $g=10^{1}$ in subplot (c).}
\label{wave_3D_comp}
\end{figure}

As before, the form of the density distribution $n_\x = |u(\x)|^2$ again changes from a clear double Gaussian shape for a small interaction strength towards a strongly interaction broadened shape at a large interaction strength. Since we deal with an anisotropic trapping potential the width in the $z$-direction is smaller than in the $y$-direction, which in turn is dominated by the $x$-direction. In figure \ref{wave_3D_comp} we illustrate this behavior by depicting cuts through the density distribution along the $x$-axis for $y,z=0$ (black lines), along the $y$-axis for $x,z=0$ (blue lines) and along the $z$-axis for $x,y=0$ (red lines). The plots are for $N_x=2^{7}$ and interaction strengths of $g=10^{-1}$, $10^{0}$ and $10^1$.

\newpage

%%%%%%%%%%%%%%%%%%%%%%%%%%%%%%%%%%%%%%%%%%%%%%%%%%%%%
\section*{Conclusion} 
%%%%%%%%%%%%%%%%%%%%%%%%%%%%%%%%%%%%%%%%%%%%%%%%%%%%%

In this paper we performed a detailed study of the minimization of the Gross-Pitaevskii functional using the method of Sobolev gradients and the trajectory given in equation \eqref{sd}.  In the analytical part of our work, we obtained global existence and uniqueness for this trajectory as well as global convergence to a minimizer of the Gross-Pitaevskii functional.  Furthermore, in our numerical part we used finite differences to discretize this trajectory and were able to find stationary solutions in one, two, and three dimensions. In our study we found that the advantages our method presents are high numerical stability, fast convergence to the desired solution, and versatility in handling various parameters such as the trapping potential, the coupling constant and initial estimates.  The main new contribution of the analysis we performed is that we give an explicit minimizing sequence that converges to the minimizer of the Gross-Pitaevskii functional. In contrast to previous comparisons of numerical methods for the minimization of this functional, our goal is not to find the smallest computation time for a given potential with well known starting point. Here we show that using an initial random field, that is without knowing anything about the final ground state, we arrive at the desired solution of the stationary Gross-Pitaevskii equation without any additional technical tricks. This is due to the global convergence of our minimizing sequence and of particular interest for arbitrary external potentials that are not as well known as the standard choices, like a harmonic potential.

The work we present in this paper opens the door for many interesting studies of the Gross-Pitaevskii energy and equation.  As a follow up numerical project, we are currently studying the performance of this method for the case of the Gross-Pitaevskii energy with rotation, which corresponds to a BEC in a rotating frame. We plan on numerically investigating the formation of vortices, but will also consider analytical studies about existence and convergence properties of vortex lattices. On the other hand, we are currently working on the adaption of our scheme to study the time dependent Gross-Pitaevskii equation. General ideas of how such an investigation can be done are presented in \cite{jwnconsv}.

Apart from physical applications our method also provides a starting point for mathematical investigations, in particular from the perspective of nonlinear semigroup theory. Defining the operator $T$ by $T(t)x = z(t)$ where $z$ is given in \eqref{sd} with $z(0)=x$, we note that $T$ is a strongly continuous nonlinear semigroup. We propose to look for the Lie generator of this group as well as studying properties of it. Since the Sobolev gradient provides an explicit construction of the trajectory with global convergence, we expect various interesting properties of this generator.

%%%%%%%%%%%%%%%%%%%%%%%%%%%%%%%%%%%%%%%%%%%%%%%%%%%%%
\section*{Acknowledgments} 
%%%%%%%%%%%%%%%%%%%%%%%%%%%%%%%%%%%%%%%%%%%%%%%%%%%%%

We thank J. W. Neuberger, R. J. Renka, E. Kajari, R. Walser and W. P. Schleich for fruitful discussions and support in numerical details. Furthermore we acknowledge financial support by the Centre National de la Recherche Scientifique (P.K.) and support by the German Space Agency DLR with funds provided by the Federal Ministry of Economics and Technology (BMWi) under grant number DLR 50 WM 0837 (M.E.).

\appendix

%%%%%%%%%%%%%%%%%%%%%%%%%%%%%%%%%%%%%%%%%%%%%%%%%%%%%
\section{Explicit representation of the projection $P_u$} 
\label{projection}
%%%%%%%%%%%%%%%%%%%%%%%%%%%%%%%%%%%%%%%%%%%%%%%%%%%%%

In this section we first provide the general setting which is used to derive a formula for $P_u$ and thereafter derive the explicit representation of $P_u$. A general result from functional analysis states that if $Q$ is a linear transformation from one Hilbert space $X$ to another Hilbert space $Y$ and if $Q^*$, the adjoint of $Q$, as a continuous linear transformation from $Y$ to $X$ has closed range, then $X$ has a unique decomposition as $X = null(Q) \oplus range(Q^*)$. Furthermore if $P$ is the orthogonal projection of $X$ onto $null(Q)$, then $I-P$ is the orthogonal projection of $X$ onto the $range(Q^*)$.  Observe that $Q^*(QQ^*)^{-1}Q$ is symmetric from $X$ to $X$, idempotent, has range in $range(Q^*)$, and is fixed on this set.  Thus $Q^*(QQ^*)^{-1}Q$ is the orthogonal projection of $X$ onto $range(Q^*)$.  This makes $I-Q^*(QQ^*)^{-1}Q$ the orthogonal projection of $X$ onto $null(Q)$. 

Now we apply this general setting to our case with $Q_u=\beta'(u)$. Then for $h \in H = H^{1,2}(D,\C)$ we obtain
\BEW
Q_u h = \beta'(u)h = 2 \Re\langle h , u \rangle_{L^2}.
\EEW
We now want to compute the adjoint of $Q_u$ as a continuous linear transformation from $H$ to $\R \subset \C$.  Recall the definition of $M$ from Definition \ref{M}.  For $r \in \R$ and $h \in H$ 
\BEW
\langle Q_uh, r \rangle_{\R}=Q_uh*r=  2\Re\langle h , u \rangle_{L^2}*r =  \langle h ,2 r M  u \rangle_H .
\EEW
Thus we see that $Q_u^* r = 2 r M u$.  Now suppose the sequence $y_n = Q_u^* r_n = 2 r_n Mu$ is in the range of $Q_u^*$ and converges to $v \in H$.  Then this sequence is Cauchy and for $\epsilon > 0$, there is $N$ so that if $m,n \ \geq \ N$
\BEW
\|2(r_n - r_m)Mu\|_H =2|r_n - r_m| \|Mu\|_H < \epsilon \, .
\EEW
This implies that $\{r_n\}_{n \geq 1}$ is a Cauchy sequence in $\R$ and hence converges to some $r \in \R$.  Thus $\{2r_n Mu \}_{ n \geq 1}$ converges in $H$ to $2r Mu$ and $v=2r Mu $ is in the range of $Q_u^*$, which implies that the range of $Q_u^*$ is closed.  Thus the orthogonal projection of $H$ onto the nullspace of $Q_u$ is given by
\BEW
P_u=I - Q_u^*(Q_uQ_u^*)^{-1}Q_u \, .
\EEW
Now observe that $Q_uQ_u^* r = Q_u(2r Mu)$ and since $Q_u h = 2\Re\langle h , u \rangle_{L^2}$ for $h \in H$, one has
\BEW
Q_u(2rMu) = 2\Re\langle u , 2r Mu \rangle_{L^2}  = 4\Re\langle u , Mu \rangle_{L^2}* r \, .
\EEW
Thus  $Q_u Q_u^* r = 4\Re\langle u , Mu \rangle_{L^2} *r$ and  $(Q_u Q_u^*)^{-1}$ exists if $ u \neq 0$. Furthermore, $(Q_u Q_u^*)^{-1}$ is continuous and given by 
\BEW
(Q_u Q_u^*)^{-1}r= \frac{1}{4\Re\langle u , Mu \rangle_{L^2}}*r \, .
\EEW
Consequently we arrive at
\begin{eqnarray*}
Q_u^*(Q_uQ_u^*)^{-1}Q_u h &=& Q_u^*(Q_uQ_u^*)^{-1} 2\Re\langle u , h \rangle_{L^2} \\
&=& Q_u^* \frac{\Re\langle u , h \rangle_{L^2}}{2\Re\langle u , Mu \rangle_{L^2}} = \frac{\Re\langle u , h \rangle_{L^2}}{\Re\langle u , Mu \rangle_{L^2}} Mu
\end{eqnarray*}
and obtain the explicit representation of the projection $P_u$ as
\BEW
P_u h&=& \left(I - Q_u^*(Q_uQ_u^*)^{-1}Q_u\right) h = h - \frac{\Re\langle u , h \rangle_{L^2}}{\Re\langle u , Mu \rangle_{L^2}} Mu \, .
\EEW

%%%%%%%%%%%%%%%%%%%%%%%%%%%%%%%%%%%%%%%%%%%%%%%%%%%%%
\section{Lipschitz property of the projection $P_u$} 
\label{lipschitz}
%%%%%%%%%%%%%%%%%%%%%%%%%%%%%%%%%%%%%%%%%%%%%%%%%%%%%

\begin{propositiona} 
Suppose $\{u_n\}_{n \geq 1}$ is a sequence of members of $H$ that converges in $K$ to $u \neq 0 \in H$.  Then the sequence $\{P_{u_n} \}_{n \geq 1}$ converges in $L(H,H)$ to $P_u$.  Furthermore there is a constant $m$ so that $\|P_{u_n} - P_u \| \leq m\|u_n - u \|_{L^2}$.
\end{propositiona} 

\begin{proof}
The proof will be given in two steps.  Let $h \in H$ with $\|h\|_H=1$.  First note that 
\BEW
\langle Mu , u \rangle_{L^2} = \langle Mu , Mu \rangle_H = \|Mu \|_H^2 \, ,
\EEW
which yields that $Mu \neq 0$ since $u \neq 0$ and $M$ is injective. Furthermore, we recall that $M \in L(X,Y)$ where $X=H, L^2$ and $Y=H, L^2$. To minimize notation we write that $\Re \langle f , g \rangle_{L^2} = \langle f , g \rangle_K$.  
%
%%\pagebreak

Thus we have  
\BEW
&&\left| \frac{\langle u_n , h \rangle_{K}}{\|Mu_n\|_H^2} - \frac{\langle u , h \rangle_K}{\|Mu\|_H^2}\right| = 
\left| \frac{\langle u_n , h \rangle_K\|Mu\|_H^2 -\langle u , h \rangle_K \|Mu_n\|_H^2}{\|Mu\|_H^2\|Mu_n\|_H^2}\right| \\
&&\quad \leq   \frac{\left|(\langle u_n , h \rangle_K-\langle u , h \rangle_K)\|Mu\|_H^2 \right|+{|\langle u , h \rangle_K (\|Mu\|_H^2 - \|Mu_n\|_H^2)|}}{\|Mu\|_H^2\|Mu_n\|_H^2}
\\
&&\quad \leq  \frac{\|u-u_n\|_{L^2}\|h\|_{L^2}\|Mu\|_H^2 + |\langle u , h \rangle_K (\|Mu\|_H^2 - \|Mu_n\|_H^2)|}{\|Mu\|_H^2\|Mu_n\|_H^2} \\
&&\quad \leq 
 \frac{\|u - u_n \|_{L^2} \|h \|_{L^2}\|u\|^2_{L^2} + \|u\|_{L^2} \|h\|_{L^2} |(\|Mu\|_H^2 - \|Mu_n\|_H^2)|}{\|Mu\|_H^2\|Mu_n\|_H^2} \\
&&\quad \leq 
\frac{\|u - u_n \|_{L^2} \|u\|^2_{L^2} + \|u\|_{L^2}  |(\|Mu\|_H^2 - \|Mu_n\|_H^2)|}{m_1}
\EEW
for some number $m_1$. Additionally, we note that 
\begin{eqnarray*}
\big|\|Mu\|_H^2 - \|Mu_n\|_H^2\big| &=& |\langle Mu , u \rangle_K - \langle Mu_n, u_n \rangle_K| \\
& \leq &
| \langle Mu , u - u_n \rangle_K| + |\langle Mu - Mu_n , u_n \rangle_K|   \\
& \leq &
\|Mu \|_{L^2} \|u - u_n \|_{L^2} + \|M(u_n - u) \|_{L^2} \|u_n\|_{L^2} \\
& \leq & m_2 \|u_n - u\|_{L^2}
\end{eqnarray*}
for some number $m_2$.  Now if we let
\BEW
c_n=\frac{\langle u_n , h \rangle_K}{\langle u_n , Mu_n \rangle_K} \quad \text{and} \quad c=\frac{\langle u , h \rangle_K}{\langle u , Mu \rangle_K} \, ,
\EEW
it is clear that there is a constant $m_3$ so that $|c_n - c| \leq m_3 \|u_n - u\|_{L^2}$. Thus we obtain 
\begin{eqnarray*}
\|(P_{u_n} - P_u)h \|_H &=& \|c_n Mu_n - c Mu \|_H \\
& \leq &
\|c_nMu_n - c_n Mu \|_H + \|c_n Mu - c Mu \|_H \\
& \leq &
|c_n| \|M(u_n-u)\|_H + |c_n - c|\|Mu\|_H
\end{eqnarray*}
and consequently there is a constant $m$ so that 
\BEW
\|(P_{u_n} - P_u)h\|_H &\leq& m \|u-u_n\|_{L^2} \, ,
\EEW
which concludes the proof.
\end{proof}

%%%%%%%%%%%%%%%%%%%%%%%%%%%%%%%%%%%%%%%%%%%%%%%%%%%%%
\section{Numerical implementation} 
%%%%%%%%%%%%%%%%%%%%%%%%%%%%%%%%%%%%%%%%%%%%%%%%%%%%%

We discretize the continuous problem in position space and for convenience we only consider the discretization of a one-dimensional problem, since the extension to higher dimensions is self-explanatory. The discrete form of the spatial variable $x$, restricted to the finite interval $[-L_x,L_x]$, reads
\BEW
\label{disc_x}
\hspace{-5mm}x = (x_1, \ldots, x_{N_x}) & \text{\;\;where\;\;} & x_n = 2L_x\Big(-\frac{1}{2}+ \frac{n-1}{N_x}\Big), \; 1 \leq n \leq N_x, \; N_x \in \N \; .
\EEW
The length $L_x$ has to be sufficiently large such that the wave function that we are interested in is numerically zero outside the interval $[-L_x,L_x]$. The existence of such a length is guaranteed by our requirement that the external potential diverges for $|\x|\to \infty$. The number of grid points $N_x$ is a very crucial parameter for the simulation because it is related to the grid spacing \mbox{$\Delta_x = 2L_x/N_x$} and thus determines the possible resolution of the scalar function $u$. The wave function $u$ on the grid introduced above is also represented as a vector
\BEW
u = (u_1, \ldots, u_{N_x}) & \text{\;\;where\;\;} & u_n = u(x_n), \; 1 \leq n \leq N_x \; .
\EEW

We approximate the spatial derivatives of the wave function via first order central differencing, which provides an accuracy for the first derivative approximation at cell centers of order ${\Delta_x}^2$. Denoting the cell centers by $e_i = (x_{i+1}+x_i)/2$ we obtain
\BEW
f'(e_i)= \frac{f_{i+1} - f_i}{\Delta_x} +O({\Delta_x}^2)\, .
\EEW
For this scheme the discretized version of the operator $W$ from the analytical part of this contribution is an $N_x-1 \times N_x$ matrix that reads
\BEW
W \; = \; \frac{1}{\Delta_x}\left(\!\!\begin{array}{rrrrrr}
-1 & 1 & 0 & \ldots & \ldots & 0 \\
0 & -1 & 1 & 0 & \ldots & 0 \\
\vdots & \ddots & \ddots & \ddots & \ddots & \vdots \\
0 & \ldots & 0 & -1 & 1 & 0\\
0 & \ldots & \ldots & 0 & -1 & 1\\
\end{array}\!\!\right) & \text{\;\;such that\;\;} & W(f) = \left(\!\!\begin{array}{c}\frac{f_2-f_1}{\Delta_x} \\ \vdots \\ \frac{f_{N_x}-f_{N_x-1}}{\Delta_x} \end{array}\!\!\right) \; .
\EEW
%\pagebreak

It is important to recall that this approximation converges towards the exact derivative in the limit of an infinitely fine grid. We use different values for the number of grid points $N_x$ which allows us to show that the Sobolev gradient indeed converges. 

Using this differencing, we see that the discrete inner product for $H^{1,2}(D)$, denoted by $\langle \cdot, \cdot \rangle_S$ is given by the following.  For $f,g$ being $\R^N$ valued functions
\BEW
\langle f , g \rangle_S= \langle f , g \rangle_N + \langle W(f) , W(g) \rangle_{N-1}
\EEW
where $\langle \cdot, \cdot \rangle_N$ denotes the $\R^N$ inner product.  Note that 
\BEW
\langle f , g \rangle_S= \langle f , g \rangle_N + \langle f , W^*W (g) \rangle_N = \langle f , (I + W^*W)(g) \rangle_N.
\EEW
$I+W^*W$ is positive definite and hence injective. Furthermore, it is invertible which allows us to obtain a result analogous to the infinite dimensional case for the relationship between the $H^{1,2}$ inner product and the $L^2$ inner product.  We see that
\BEW
\langle f , (I+W^*W)^{-1} g \rangle_S =  \langle f , g \rangle_N 
\EEW
and thus the analogous finite dimensional $M$ as given in the first part is $(I+W^*W)^{-1}$.  

Now, all we need for the numerical simulation is an easily accessible procedure to calculate the {Sobolev} gradient. We restrict ourselves to real-valued functions $u$, since it can be shown, that the ground state can be chosen to be real-valued. The Fr\'{e}chet derivative of the energy functional then reads
\BEW
E'(u)h &=& \int_D \big(\Re \langle \nabla u(\x), \nabla h(\x) \rangle + 2V_{trap}(\x) \Re\langle u(\x),h(\x) \rangle \\
&&\phantom{\int_D \big(} + 2g |u(\x)|^2 \Re \langle u(\x), h(\x) \rangle\big) \, \d\x  \,.
\EEW
where $\langle f, g \rangle$ is a short-hand notation for the product of $f$ and the complex conjugate of $g$. By adding and subtracting $\int_D \Re \langle u(\x), h(\x) \rangle \, \d\x$ from this term, one obtains
\BEW
E'(u)h = \Re \langle h , u \rangle_H + \Re \langle h ,  2V_{trap} u +  2g |u|^2 u - u \rangle_{L^2} \, . 
\EEW
Making use of the previously mentioned relationship between the $H^{1,2}$ inner product and the $L^2$ inner product we arrive at
\BEW
E'(u)h &=& \Re \langle h , u + M ( 2 V_{trap} u + 2g |u|^2 u - u) \rangle_H \, 
\EEW
where $M$ is as in Definiton \ref{M}.
%\pagebreak
Recall that the Sobolev gradient of $E$ at $u$, $\nabla_H E(u)$, was defined to be the element of $H$ so that 
\BEW
E'(u)h = \langle h , \nabla_H E(u) \rangle_H = \Re \langle h , \nabla_H E(u) \rangle_H
\EEW
as $E'(u)h$ is real valued.  Therefore, the Sobolev gradient of $E$ at $u$ is given by 
\BEW
\nabla_H E(u) = u + M ( 2V_{trap} u + 2g |u|^2 u - u ) \, .
\EEW
Incorporating the projection of $H$ onto the nullspace of $\beta'(u)$ one sees that 
\BEW
P_u \nabla_H E(u) = \nabla_H E(u) - \frac{\Re \langle u , \nabla_H E(u) \rangle_{L^2}}{ \Re\langle u , Mu \rangle_{L^2}} Mu \, .
\EEW

As soon as the minimization of the energy functional has converged in the sense of the relative change in energy from one iteration step to the next is less than $10^{-4}$, we switch to the well-known Newton method in order to find the exact solution of the stationary Gross-Pitaevskii equation
\BEW
- \frac{1}{2} \nabla^2 u(\x) + V_{trap}(\x) u(\x) + g |u(\x)|^2u(\x) = \mu u(\x)
\EEW
with the chemical potential potential $\mu$. Hence it only remains to show, how the calculation of the chemical potential is implemented. For this purpose, we multiply both sides of this equation by the complex conjugate of $u$ and integrate over the domain $D$. Thus, we arrive at
\BEW
\mu = \frac{1}{N} \int_D \left(\frac{ |\nabla u(\x)|^2}{2} + V_{trap}(\x) |u(\x)|^2 + g |u(\x)|^4\right)\d\x \, .
\EEW

\newpage

\bibliographystyle{elsart-num}
\bibliography{References}

\end{document}